\documentclass[a4paper]{IEEEtran}
\usepackage{cite}
\ifCLASSINFOpdf
\else
\fi
\usepackage[cmex10]{amsmath}
\usepackage{array}
\ifCLASSOPTIONcompsoc
  \usepackage[caption=false,font=normalsize,labelfont=sf,textfont=sf]{subfig}
\else
  \usepackage[caption=false,font=footnotesize]{subfig}
\fi
\usepackage{url}
\usepackage{tabulary}
\usepackage[english]{babel}
\usepackage{blindtext}
\usepackage[makeroom]{cancel}
\usepackage{amsfonts}
\usepackage{amsthm}  %This is for proofs
\usepackage{amssymb}
\usepackage[pdftex]{graphicx}
\graphicspath{{../pdf/}{../jpeg/}}
\DeclareGraphicsExtensions{.pdf,.jpeg,.png}
\usepackage{pgf,tikz}
\usepackage{tabularx}
\usepackage{float}

\usepackage{enumitem} % To reference numbered items

\newcommand{\vect}[1]{\boldsymbol{#1}}
\newcommand{\mat}[1]{\boldsymbol{#1}}
\newcommand{\wt}[1]{\widetilde{#1}}
\newcommand{\wh}[1]{\widehat{#1}}

\newtheorem{remark}{Remark}
\newtheorem{theorem}{Theorem}
\newtheorem{lemma}{Lemma}

\newtheorem{definition}{Definition}
\newtheorem{corollary}{Corollary}

\newtheorem{property}{Property}
\newtheorem{assumption}{Assumption}
\newtheorem{fact}{Fact}
\newtheorem{objective}{Objective}

\begin{document}

\title{Evolution of Social Power in Social Networks with Dynamic Topology}

\author{\IEEEauthorblockN{Mengbin Ye, \emph{Student Member, IEEE} $\quad$ }
\and
\IEEEauthorblockN{Ji Liu, \emph{Member, IEEE}$\quad $ }
\and 
\IEEEauthorblockN{Brian D.O. Anderson, \emph{Life Fellow, IEEE} $\quad$ }
\and \\
\IEEEauthorblockN{Changbin Yu, \emph{Senior Member, IEEE $\quad$ }
\and
\IEEEauthorblockN{Tamer Ba\c{s}ar, \emph{Life Fellow, IEEE} }
}

\thanks{M.~Ye, B.D.O.~Anderson and C.~Yu are with the Research School of Engineering, Australian National University
\texttt{\{Mengbin.Ye, Brian.Anderson, Brad.Yu\}@anu.edu.au}. B.D.O.~Anderson is also with Hangzhou Dianzi University, Hangzhou, China, and with Data61-CSIRO (formerly NICTA Ltd.) in Canberra, A.C.T., Australia. J.~Liu and T.~Ba\c{s}ar are with the Coordinated Science
Laboratory, University of Illinois at Urbana-Champaign
\texttt{\{jiliu, basar1\}@illinois.edu}.}% <-this % stops a space

}

\maketitle

\begin{abstract}
The recently proposed DeGroot-Friedkin model describes the dynamical evolution of individual social power in a social network that holds opinion discussions on a sequence of different issues. This paper revisits that model, and uses nonlinear contraction analysis, among other tools, to establish several novel results. First, we show that for a social network with constant topology, each individual's social power converges to its equilibrium value exponentially fast, whereas previous results only concluded asymptotic convergence. Second, when the network topology is dynamic (i.e., the relative interaction matrix may change between any two successive issues), we show that each individual exponentially forgets its initial social power. Specifically, individual social power is dependent only on the dynamic network topology, and initial (or perceived) social power is forgotten as a result of sequential opinion discussion. Last, we provide an explicit upper bound on an individual's social power as the number of issues discussed tends to infinity; this bound depends only on the network topology. Simulations are provided to illustrate our results.
\end{abstract}


% \fi
\IEEEpeerreviewmaketitle

\section{Introduction}\label{sec:intro}

\IEEEPARstart{S}{ocial} network analysis is the study of a group of social actors (individuals or organisations) who interact in some way according to a social connection or relationship. The study of social networks has spanned several decades \cite{cartwright1959social_book,friedkin2015_socialsurvey} and across several scientific communities. In the past few years, perhaps in part due to lessons learned and tools developed from extensive research on coordination of autonomous multi-agent systems \cite{murray}, the systems and control community has taken an interest in social network analysis.

Of particular interest in this context is the problem of ``opinion dynamics'', which is the study of how individuals in a social network interact and exchange their opinions on an issue or topic. A critical aspect is to develop models which simultaneously capture observed social phenomena and are simple enough to be analysed, particularly from a system-theoretic point of view. The seminal works of \cite{french1956_socialpower,degroot1974OpinionDynamics} proposed a discrete-time opinion pooling/updating rule, now known as the French-DeGroot (or simply DeGroot) model. A continuous-time counterpart, known as the Abelson model, was proposed in \cite{abelson1964op_dyn}. These opinion updating rules are closely related to consensus algorithms for coordinating autonomous multi-agent systems \cite{jadbabaie2003_CoordinationAgents,ren2005consensus}. The Friedkin-Johnsen model \cite{friedkin1990_FJsocialmodel,friedkin2006structural_theory_book} extended the French-DeGroot model by introducing the concept of a ``stubborn individual'', i.e., an individual who remains attached to its initial opinion. This helped to model \emph{social cleavage}\cite{friedkin2015_socialsurvey}, a phenomenon where opinions tend towards separate clusters. Other models which attempt to explain social cleavage include the Altafini model with negative/antagonistic interactions \cite{altafini2013antagonistic_interactions,altafini2015predictable_opinions,proskurnikov2016opinion,liu2017ex} and the Hegelsmann-Krause bounded confidence model \cite{hegselmann2002opinion,etesami2015game}. Simultaneous opinion discussion on multiple, logically interdependent topics was studied with a multidimensional Friedkin-Johnsen model \cite{parsegov2017_multiissue,friedkin2016network_science}. 

The concept of \emph{social power} or \emph{social influence} has been integral throughout the development of these models. Indeed, French Jr's seminal paper \cite{french1956_socialpower} was an attempt to quantitatively study an individual's social power in a group discussion. Broadly speaking, in the context of opinion dynamics, individual social power is the amount of influence an individual has on the overall opinion discussion. Individuals which maximise the spread of an idea or rumour in diffusion models were identified in \cite{kempe2003_maxspread}. The social power of an individual in a group can change over time as group members interact and are influenced by each other. Recently, the DeGroot-Friedkin model was proposed in \cite{jia2015opinion_SIAM} to study the dynamic evolution of an individual's social power as a social network discusses opinions on a sequence of issues. In this paper, we present several major, novel results on the DeGroot-Friedkin model. In Section \ref{sec:background}, we shall provide a precise mathematical formulation of the model, but here we provide a brief description to better motivate the study, and elucidate the contributions of the paper.

The discrete-time DeGroot-Friedkin model \cite{jia2015opinion_SIAM} is a two-stage model. 
In the first stage, individuals update their opinions on a particular issue, and in the second stage, each individual's level of self-confidence for the next issue is updated.
For a given issue, the social network discusses opinions using the DeGroot opinion updating model, which has been empirically shown to outperform Bayesian learning methods in the modelling of social learning processes \cite{chandrasekhar2012testing}. The row-stochastic opinion update matrix
%(i.e. every entry is nonnegative and each row sum is equal to 1) 
used in the DeGroot model is parametrised by two sets of variables. The first is individual social powers, which are the diagonal entries of the opinion update matrix (i.e. the weight an individual places on its own opinion). The second is the relative interaction matrix, which is used to scale the off-diagonal entries of the opinion update matrix to ensure that, for any given values of individual social powers, the opinion update matrix remains row-stochastic. In the original model \cite{jia2015opinion_SIAM}, the relative interaction matrix was assumed to be constant over all issues, and constant throughout the opinion discussion on any given issue. Under some mild conditions on the entries of the relative interaction matrix, the opinions reach a consensus on every issue.

At the end of the period of discussion of an issue, i.e., when opinions have effectively reached a consensus, each individual undergoes a sociological process of \emph{self-appraisal} (detailed in the seminal work \cite{cooley1992human_nature}) to determine its impact or influence on the final consensus value of opinion. Such a mechanism is well accepted as a hypothesis \cite{shrauger1979symbolic_selfappraisal,gecas1983self_appraisal} and has been empirically validated \cite{yeung2003selfappraisal_empirical}. Immediately before discussion on the next issue, each individual self-appraises and updates its individual social power (the weight an individual places on its own opinion) according to the impact or influence it had on discussion of the previous issue.	In updating its individual social power, an individual also updates the weight it accords its neighbours' opinions, by scaling using the relative interaction matrix, to ensure that the opinion updating matrix for the next issue remains row-stochastic. This process is repeated as issues are discussed in sequence. The primary objective of the DeGroot-Friedkin model is to \emph{study the dynamical evolution of the individual social powers over the sequence of discussed issues.}

The model is centralised in the sense that individuals are able to observe and detect their impact relative to every other individual in the opinion discussions process, which indicates that the DeGroot-Friedkin model is best suited for networks of small or moderate size. Such networks are found in many decision making groups such as boards of directors, government cabinets or jury panels. Distributed models of self-appraisal have been studied in continuous time \cite{chen2017_DFdistributed} as well as discrete time \cite{xu2015_modified_DF, xia2016_modified_DF_timevary} to extend the original DeGroot-Friedkin model. Dynamic topology, but restricted to doubly-stochastic relative interaction matrices, was studied in \cite{xia2016_modified_DF_timevary}.

\subsection{Contributions of This Paper}
This paper significantly expands on the original DeGroot-Friedkin model in several different respects. In the original paper \cite{jia2015opinion_SIAM}, LaSalle's Invariance Principle was used to arrive at an asymptotic stability result. Exponential convergence was conjectured but not proved. In this paper, a novel approach based on nonlinear contraction analysis \cite{lohmiller1998contraction} is used to conclude an exponential convergence property for non-autocratic social power configurations. Autocratic social power configurations are shown to be unstable, or asymptotically stable, but not exponentially so. Additional insights are also developed; an upper bound on the individual social power at equilibrium is established, dependent only on the relative interaction matrix. The ordering of individuals' equilibrium social powers can be determined \cite{jia2015opinion_SIAM}, but numerical values for nongeneric network topologies cannot be determined.

The paper is also the first to provide a complete proof of convergence for the DeGroot-Friedkin model with \emph{dynamic topology}. Dynamic topology for the DeGroot-Friedkin model was studied in \cite{friedkin2016tevo_power} and a stability result was conjectured based on extensive simulation. By dynamic topology, we mean relative interaction matrices which are different between issues, \emph{but remain constant during the period of discussion for any given issue}. Relative interaction matrices encode trust or relationship strength between individuals in a network. A network discussing sometimes sports and sometimes politics will have different interaction matrices; some individuals are experts on sports and others on politics. These factors can influence the trust or relationship strength between individuals. This gives rise to the concept of \emph{issue-driven} topology change. In addition, allowing for dynamic relative interaction matrices is a natural way of describing \emph{network structural changes over time}. For many reasons, new relationships may form and others may die out. For example, an individual may attempt to, after each issue, form new relationships, disrupt other relationships, and adjust relationship strengths in order to maximise its individual social power. This gives rise to the concept of \emph{individual-driven} topology change. The idea that an individual intentionally modifies topology to gain its social power was studied in \cite{ye2017socialpower_mod} by assuming constant topology, but this can be more naturally modelled using dynamic topology.

A conference paper \cite{ye2017DF_IFAC} by the authors studied the special case of periodically varying topology and proved the existence of periodic trajectories, but did not provide a convergence proof. In this paper, we show that for relative interaction matrices which vary arbitrarily across issues, the individual social powers converge exponentially fast to a unique trajectory (as opposed to unique stationary values for constant interactions). Specifically, every individual forgets its initial social power estimate (initial condition) for each issue exponentially fast. For any given issue, and as the number of issues discussed tends to infinity, individuals' social powers are determined only by the network interactions on the previous issue. This paper therefore concludes that a social network described by the DeGroot-Friedkin model is {\em self-regulating} in the sense that, even on dynamic topologies, sequential discussion combined with reflected self-appraisal removes perceived social power (initial estimates of social power). True social power is determined by topology. Periodically varying topologies are presented as a special case.

\subsection{Structure of the Rest of the Paper}
Section~\ref{sec:background} introduces mathematical notations, nonlinear contraction analysis and the DeGroot-Friedkin model. Section~\ref{sec:constant_C} uses nonlinear contraction analysis to study the original DeGroot-Friedkin model. Dynamic topologies are studied in Section~\ref{sec:dyn_top}. Simulations are presented in Section~\ref{sec:sim}, and concluding remarks are given in Section~\ref{sec:conclusion}.

\section{Background and Problem Statement}\label{sec:background}
We begin by introducing some mathematical notations used in the paper. Let $\vect 1_n$ and $\vect 0_n$ denote, respectively, the $n\times 1$ column vectors of all ones and all zeros. For a vector $\vect x\in\mathbb{R}^n$, $0\preceq\vect x$ and $0 \prec \vect x$ indicate component-wise inequalities, i.e., for all $i\in\{1,\ldots,n\}$, $0\leq x_i$ and $0<x_i$, respectively. The $n$-simplex is \mbox{$\Delta_n = \{\vect x\in \mathbb{R}^n : 0 \preceq \vect x, \vect 1_n^\top \vect x = 1 \}$}. The canonical basis of $\mathbb{R}^n$ is given by $\mathbf{e}_1, \ldots, \mathbf{e}_n$. Define $\wt{\Delta}_n = \Delta_n \backslash \{ \mathbf{e}_1, \ldots, \mathbf{e}_n \}$ and \mbox{$\text{int}(\Delta_n) = \{\vect x\in \mathbb{R}^n : 0 \prec \vect x, \vect 1_n^\top \vect x = 1 \}$}. The $1$-norm and infinity-norm of a vector, and their induced matrix norms, are denoted by $\Vert \cdot\Vert_1$ and $\Vert \cdot \Vert_{\infty}$, respectively. For the rest of the paper, we shall use the terms ``node'', ``agent'', and ``individual'' interchangeably. We shall also interchangeably use the words ``self-weight'', ``social power'', and ``individual social power''. 

An $n\times n$ matrix with all entries nonnegative is called a {\em row-stochastic matrix} (respectively \emph{doubly stochastic}) if its row sums all equal 1 (respectively if its row and column sums all equal 1). We now provide a result on eigenvalues of a matrix product, to be used later.

\begin{lemma}[Corollary 7.6.2 in \cite{horn2012matrixbook}]\label{cor:AB_real}
Let $\mat{A}, \mat{B} \in \mathbb{R}^{n\times n}$ be symmetric. If $\mat{A}$ is positive definite, then $\mat{AB}$ is diagonalizable and has real eigenvalues. If, in addition, $\mat{B}$ is positive definite or positive semidefinite, then the eigenvalues of $\mat{AB}$ are all strictly positive or nonnegative, respectively.
\end{lemma}

\subsection{Graph Theory}\label{ssec:graph_theory}
The interaction between individuals in a social network is modelled using a weighted directed graph, denoted as $\mathcal{G} = (\mathcal{V}, \mathcal{E}, \mat{C})$. Each individual corresponds to a node in the finite, nonempty set of nodes $\mathcal{V} = \{v_1, \ldots, v_n\}$. The set of ordered edges is $\mathcal{E} \subseteq \mathcal{V}\times \mathcal{V}$. We denote an ordered edge as $e_{ij} = (v_i, v_j) \in \mathcal{E}$, and because the graph is directed, in general, $e_{ij}$ and $e_{ji}$ may not both exist. An edge $e_{ij}$ is said to be outgoing with respect to $v_i$ and incoming with respect to $v_j$. The presence of an edge $e_{ij}$ connotes that individual $j$ learns of, and takes into account, the opinion value of individual $i$  when updating its own opinion. The incoming and outgoing neighbour sets of $v_i$ are respectively defined as $\mathcal{N}_i^+ = \{v_j \in \mathcal{V} : e_{ji} \in \mathcal{E}\}$ and $\mathcal{N}_i^- = \{v_j \in \mathcal{V} : e_{ij} \in \mathcal{E}\}$. The relative interaction matrix $\mat C\in\mathbb{R}^{n\times n}$ is associated with $\mathcal{G}$, the relevance of which is explained below. The matrix $\mat{C}$ has nonnegative entries $c_{ij}$, termed ``relative interpersonal weights'' in \cite{jia2015opinion_SIAM}.  The entries of $\mat C$ have properties such that $0 < c_{ij} \leq 1 \Leftrightarrow e_{ji} \in \mathcal{E}$ and $c_{ij} = 0$ otherwise. It is assumed that $c_{ii} = 0$ (i.e., there are no self-loops), and we impose the restriction that $\sum_{j\in\mathcal{N}_i^+} c_{ij} = 1$ (i.e., $\mat C$ is a row-stochastic matrix). The word ``relative'' therefore refers to the fact that $c_{ij}$ can be considered as a percentage of the total weight or trust individual $i$ places on individual $j$ compared to all of individual $i$'s incoming neighbours.

A directed path is a sequence of edges of the form $(v_{p_1}, v_{p_2}), (v_{p_2}, v_{p_3}), \ldots$ where $v_{p_i} \in \mathcal{V}, e_{ij} \in \mathcal{E}$. Node $i$ is reachable from node $j$ if there exists a directed path from $v_j$ to $v_i$. A graph is said to be strongly connected if every node is reachable from every other node. The relative interaction matrix $\mat C$ is irreducible if and only if the associated graph $\mathcal{G}$ is strongly connected. If $\mat C$ is irreducible, then it has a unique left eigenvector $\vect{\gamma}^\top \succ 0$ satisfying $\vect{\gamma}^\top \vect{1}_n = 1$, associated with the eigenvalue 1 (Perron-Frobenius Theorem, see \cite{godsil2001algebraic}). Henceforth, we call $\vect{\gamma}^\top$ the \emph{dominant left eigenvector of $\mat C$}.

\subsection{The DeGroot-Friedkin Model}\label{ssec:df_model}
We define $\mathcal{S} = \{0, 1, 2, \ldots\}$ to be the set of indices of sequential issues which are being discussed by the social network. For a given issue $s \in \mathcal{S}$, the social network discusses it using the discrete-time DeGroot consensus model (with constant weights throughout the discussion of the issue). At the end of the discussion (i.e. when the DeGroot model has effectively reached steady state), each individual undergoes reflected self-appraisal, with ``reflection'' referring to the fact that self-appraisal occurs following the completion of discussion on the particular issue $s$. Each individual then updates its own self-weight, and discussion begins on the next issue $s+1$ (using the DeGroot model but now with adjusted weights).

\begin{remark}[Time-scales]
The DeGroot-Friedkin model assumes the opinion dynamics process operates on a different time-scale than that of the reflected appraisal process. This allows for a simplification in the modelling and is reasonable if we consider that having separate time-scales merely implies that the social network reaches a consensus on opinions on one issue before beginning discussion on the next issue. If this assumption is removed, i.e., the time-scales are comparable, then the distributed DeGroot-Friedkin model is used \cite{xu2015_modified_DF}. However, at this point the analysis of the distributed model is much more involved, and has not yet reached the same level of understanding as the original model.  
\end{remark}

We next explain the mathematical modelling of the opinion dynamics for an issue and the updating of self-weights from one issue to the next. 

\subsubsection{DeGroot Consensus of Opinions}\label{sssec:degroot}
For each issue $s \in \mathcal{S}$, individual $i$ updates its opinion $y_i(s,\cdot) \in \mathbb{R}$ at time $t+1$ as
\begin{equation*}
y_i(s, t+1) = w_{ii}(s) y_i(s, t) + \sum_{j\neq i}^n w_{ij}(s) y_j(s, t)
\end{equation*}
where $w_{ii}(s)$ is the self-weight individual $i$ places on its own opinion and $w_{ij}(s)$ is the weight placed by individual $i$ on the opinion of its neighbour individual $j$. Note that $\forall\,i,j$, $w_{ij}(s) \in [0,1]$ is constant for any given $s$. As will be made apparent below, $\sum_{j = 1}^n w_{ij} = 1$, which implies that individual $i$'s new opinion value $y_i(s, t+1)$ is a convex combination of its own opinion and the opinions of its neighbours at the current time instant. The opinion dynamics for the entire social network can be expressed as 
\begin{equation}\label{eq:opinion_network}
\vect y(s, t+1) = \mat W(s) \vect y(s, t)
\end{equation}
where $\vect y(s, t) = [y_1(s, t), \, \hdots ,\, y_n(s, t)]^\top$ is the vector of opinions of the $n$ individuals in the network at time instant $t$. This model was studied in \cite{french1956_socialpower,degroot1974OpinionDynamics} with $\mathcal{S} = \{0\}$ (i.e., only one issue was discussed), and with individuals who remember their initial opinions $y_i(s,0)$ \cite{friedkin1990_FJsocialmodel,friedkin2006structural_theory_book}.

Let the self-weight (individual social power) of individual $i$ be denoted by $x_i(s) = w_{ii}(s) \in [0,1]$ (the $i^{th}$ diagonal entry of $\mat W(s)$) \cite{jia2015opinion_SIAM}, with the individual social power vector given as $\vect{x}(s) = [x_1, \hdots, x_n]^\top$. For a given issue $s$, the influence matrix $\mat W(s)$ is defined as 
\begin{equation}\label{eq:W_matrix}
\mat W(s) = \mat X(s) + (\mat I_n - \mat X(s))\mat C
\end{equation}
where $\mat C$ is the relative interaction matrix associated with the graph $\mathcal{G}$, and the matrix $\mat X(s) \doteq diag[\vect x(s)]$. From the fact that $\mat C$ is row-stochastic with zero diagonal entries, \eqref{eq:W_matrix} implies that $\mat W(s)$ is a row-stochastic matrix. It has been shown in \cite{jia2015opinion_SIAM} that $\mat{W}(s)$ defined as in \eqref{eq:W_matrix} ensures that for any given $s$, there holds $\lim_{t \to \infty} \vect y(s,t) = (\vect{\zeta}(s)^\top \vect y(s,0))\vect{1}_n$. Here, $\vect \zeta(s)^\top$ is the unique nonnegative left eigenvector of $\mat W(s)$ associated with the eigenvalue $1$, normalised such that $\vect 1_n^\top \vect \zeta(s) = 1$. That is, the opinions converge to a constant consensus value.

Next, we describe the model for the updating of $\mat W(s)$ (specifically $w_{ii}(s)$ via a reflected self-appraisal mechanism).  Kronecker products may be used if each individual has simultaneous opinions on $p$ unrelated topics, $\vect y_i \in \mathbb{R}^p, p \geq 2$. Simultaneous discussion of $p$ logically interdependent topics is treated in \cite{friedkin2016network_science,parsegov2017_multiissue} under the assumption that $\mathcal{S} = \{0\}$. 

\subsubsection{Friedkin's Self-Appraisal Model for Determining Self-Weight}\label{sssec:friedkin} 
The Friedkin component of the model proposes a method for updating the individual self-weights, $\vect{x}(s)$. We assume the starting self-weights $x_i(0)\geq 0$ satisfy $\sum_i x_i(0) = 1$.\footnote{The assumption that $\sum_i x_i(0) = 1$ is not strictly required, as we will prove in Section~\ref{sec:dyn_top} that if $0\leq x_i(0) < 1,\forall\,i$ and $\exists\,j : x_j(0) > 0$, then the system will remain inside the simplex $\Delta_n$ for all $s \geq 1$.}
At the end of the discussion of issue $s$, the self-weight vector updates as
\begin{equation}\label{eq:x_update}
\vect x(s+1) = \vect{\zeta}(s)
\end{equation}
Note that $\vect{\zeta}(s)^\top \vect{1}_n = 1$ implies that $\vect x(s) \in \Delta_n$, i.e., $\sum_{i=1}^n x_i(s) = 1$ for all $s$. From \eqref{eq:W_matrix}, and because $\mat{C}$ is row-stochastic, it is apparent that by adjusting $w_{ii}(s+1) = \zeta_i(s)$, individual $i$ also scales $w_{ij}(s+1), j \neq i$ using $c_{ij}$ to be $(1-w_{ii}(s+1))c_{ij}$ to ensure that $\mat{W}(s)$ remains row-stochastic.

\begin{remark}[Social Power]\label{rem:selfweight_update}
The precise motivation behind using \eqref{eq:x_update} as the updating model for $\vect{x}(s)$ is detailed in \cite{jia2015opinion_SIAM}, but we provide a brief overview here in the interest of making this paper self-contained. As discussed in Subsection~\ref{sssec:degroot}, for any given $s$, there holds $\lim_{t \to \infty} \vect{y}(s,t) = (\vect{\zeta}(s)^\top \vect y(s,0))\vect{1}_n$. In other words, for any given issue $s$, the opinions of every individual in the social network reaches a consensus value $\vect{\zeta}(s)^\top \vect y(s,0)$ equal to a convex combination of their initial opinion values $\vect y(s,0)$. The elements of $\vect{\zeta}(s)^\top$ are the convex combination coefficients. For a given issue $s$, $\zeta_i(s)$ is therefore a precise manifestation of individual $i$'s social power or influence in the social network, as it is a measure of the ability of individual $i$ to control the outcome of a discussion \cite{cartwright1959social_book}. The reflected self-appraisal mechanism therefore describes an individual $1)$ observing how much power it had on the discussion of issue $s$ (the nonnegative quantity $\zeta_{i}(s)$), and $2)$ for the next issue $s+1$, adjusting its self-weight to be equal to this power, i.e., $x_i(s+1) = w_{ii}(s+1) = \zeta_i(s)$.
\end{remark}

Lemma 2.2 of \cite{jia2015opinion_SIAM} showed that the system \eqref{eq:x_update} is equivalent to the discrete-time system
\begin{equation}\label{eq:DF_system}
\vect x(s+1) = \vect F(\vect x(s))
\end{equation}
where the nonlinear map $\vect F(\vect x(s))$ is defined as
\begin{align}\label{eq:map_F_DF}
\vect F( \vect x(s) ) =   \begin{cases} 
   \mathbf e_i & \hspace*{-6pt} \text{if } \vect x(s) = \mathbf e_i \; \text{for any } i \\ \\
   \alpha (\vect x(s)) \begin{bmatrix} \frac{\gamma_1}{1-x_1(s)} \\ \vdots \\ \frac{\gamma_n}{1-x_n(s)} \end{bmatrix}       & \text{otherwise }
  \end{cases}
\end{align}
with $\alpha(\vect x(s)) = 1/\sum_{i=1}^n \frac{\gamma_i}{1- x_i(s)}$ where  $\vect{\gamma} = [\gamma_1, \gamma_2, \hdots, \gamma_n]^\top$ is the dominant left eigenvector of $\mat{C}$. Note that $\sum_i F_i = 1$, where $F_i$ is the $i^{th}$ entry of $\vect{F}$. We now introduce an assumption \emph{which will be invoked throughout the paper}.

\begin{assumption}\label{assm:C_matrix}
The matrix $\mat{C} \in \mathbb{R}^{n\times n}$, with $n \geq 3$, is irreducible, row-stochastic, and has zero diagonal entries. Irreducibility of $\mat{C}$ implies, and is implied by, the strongly connectedness of the graph $\mathcal{G}$ associated with $\mat{C}$.
\end{assumption}
This assumption was in place in \cite{jia2015opinion_SIAM} by and large throughout its development. Dynamic topology involving reducible $\mat{C}$ is a planned future work of the authors. A special topology studied in \cite{jia2015opinion_SIAM} is termed ``star topology'', the definition and relevance of which follow.
\begin{definition}[Star topology]\label{def:star}
A strongly connected graph\footnote{While it is possible to have a star graph that is not strongly connected, this paper, similarly to \cite{jia2015opinion_SIAM}, deals only with strongly connected graphs.} $\mathcal{G}$ is said to have star topology if $\exists$ a node $v_i$, called the centre node, such that every edge of $\mathcal{G}$ is either to or from $v_i$.
\end{definition}
The irreducibility of $\mat C$ implies that a star $\mathcal{G}$ must include edges in both directions between the centre node $v_i$ and every other node $v_j, j \neq i$. We now provide a lemma and a theorem (the key result of \cite{jia2015opinion_SIAM}) regarding the convergence of $\vect{F}(\vect{x}(s))$ as $s\to \infty$, and a fact useful for analysis throughout the paper. 

\begin{lemma}[Lemma 3.2 in \cite{jia2015opinion_SIAM}]\label{lem:star}
Suppose that $n \geq 3$, and suppose further that $\mathcal{G}$ has star topology, which without loss of generality has centre node $v_1$. Let $\mat{C}$ satisfy Assumption~\ref{assm:C_matrix}. Then, $\forall\,\vect{x}(0) \in \wt{\Delta}_n$, $\lim_{s\to\infty} \vect x(s) = \mathbf{e}_1$.
\end{lemma}

This implies that $\forall\,\vect{x}(0) \in \wt{\Delta}_n$, a network with star topology converges to an ``autocratic configuration'' where centre individual $1$ holds all of the social power. 

\begin{fact}\cite{jia2015opinion_SIAM}\label{fact:star_gamma}
Suppose that $n\geq 3$ and let $\vect{\gamma}^\top$, with entries $\gamma_i$, be the dominant left eigenvector of $\mat{C}\in \mathbb{R}^{n\times n}$, satisfying Assumption~\ref{assm:C_matrix}. Then, $\Vert \vect{\gamma}\Vert_{\infty} = 0.5$ if and only if $\mat{C}$ is associated with a star topology graph, and in this case $\gamma_i = 0.5$ where $i$ is the centre node; otherwise,  $\Vert \vect{\gamma}\Vert_{\infty} < 0.5$.
\end{fact}

\begin{theorem}[Theorem 4.1 in \cite{jia2015opinion_SIAM}]\label{thm:DFmain}
For $n\geq 3$, consider the DeGroot-Friedkin dynamical system \eqref{eq:DF_system} with $\mat C$ satisfying Assumption~\ref{assm:C_matrix}. Assume further that the digraph $\mathcal{G}$ associated with $\mat C$ does not have star topology. Then,
\begin{enumerate}[label=(\roman*)]
\item \label{prty:thm_DFmain01} For all initial conditions $\vect x(0) \in \wt{\Delta}_n $, the self-weights $\vect x(s)$ converge to $\vect x^*$ as $s\to\infty$, where $\vect x^* \in \text{int}({\Delta}_n)$ is the unique fixed point satisfying $\vect x^* = \vect F(\vect x^*)$. 
\item \label{prty:thm_DFmain02} There holds $x^*_i < x^*_j$ if and only if $\gamma_i < \gamma_j$ for any $i,j$, where $\gamma_i$ is the $i^{th}$ entry of the dominant left eigenvector $\vect\gamma$. There holds $x^*_i = x^*_j$ if and only if $\gamma_i = \gamma_j$. 
\item \label{prty:thm_DFmain03} The unique fixed point $\vect x^*$ is determined only by $\vect\gamma$, and is independent of the initial conditions. 
\end{enumerate}
\end{theorem}

\subsection{Quantitative Aspects of the Dynamic Topology Problem}\label{ssec:problem_def}
In the introduction, we discussed in qualitative terms that we are seeking to study the evolution, and in particular the convergence properties, of social power in dynamically changing social networks. Now, we provide quantitative details on the problem of interest. Specifically, we will consider dynamic relative interaction matrices $\mat{C}(s)$ which are \emph{issue-driven or individual-driven}. As we have now properly introduced the DeGroot-Friedkin model, it is appropriate for us to expand on this motivation, using the following two examples.

\emph{Example 1 [Issue-driven]:} Consider a government cabinet that meets to discuss the issues of defence, economic growth, social security programs and foreign policy. Each minister (individual in the cabinet) has a specialist portfolio (e.g. defence) and perhaps a secondary portfolio (e.g. foreign policy). While every minister will partake in the discussion of each issue, the weights $c_{ij}(s)$ will change. For example, if minister $i$'s portfolio is on defence, then $c_{ji}(s_{\text{defence}})$ will be high as other ministers $j$ place more trust on minister $i$'s opinion. On the other hand, $c_{ji}(s_{\text{security}})$ will be low. It is then apparent that $\mat{C}(s_{\text{defence}}) \neq \mat{C}(s_{\text{security}})$ in general. This motivates the incorporation of \emph{issue-dependent or issue-driven} topology into the DeGroot-Friedkin model.

\emph{Example 2 [Individual-driven]:} Consider individual $i$ and individual $j$ in a network, and suppose that $c_{ij}(s) = 0$ for $s = 0$. However, after several discussions (say $5$ issues), individual $i$ has observed that individual $j$ consistently has a high impact on discussions, i.e., $\zeta_j(s)$ is large. Then, individual $i$ may form an interpersonal relationship such that $c_{ij}(s) > 0$ for $s \geq 6$ (which implies that individual $i$ begins to take into consideration the opinion of individual $j$). 

The two examples above are different from each other, but both equally provide motivation for \emph{dynamic} topology. We assume that $\forall\,s$, $\mat{C}(s)$ satisfies Assumption~\ref{assm:C_matrix}. Given that $\mat{C}(s)$ is dynamic, the opinion dynamics for each issue is then given by $\vect{y}(s,t+1) = \mat{W}(s)\vect{y}(s,t)$ where 
\begin{equation}\label{eq:W_s_dynamic_C}
\mat{W}(s) = \mat{X}(s) + (\mat{I}_n - \mat{X}(s))\mat{C}(s)
\end{equation}
which records the fact that $\mat{C}(s)$ is dynamic, in distinction to \eqref{eq:W_matrix}. Precise details of the adjustments to the model arising from dynamic $\mat{C}$ are left for Section~\ref{sec:dyn_top}. We can thus formulate the key objective of this paper at this point as follows.
\begin{objective}
To study the dynamic evolution (including convergence) of $\vect{x}(s)$ over a sequence of discussed issues by using the DeGroot model \eqref{eq:opinion_network} for opinion discussion, where $\mat{W}(s)$ is given in \eqref{eq:W_s_dynamic_C}, with the reflected self-appraisal mechanism \eqref{eq:x_update} used to update $\vect{x}(s)$. 
\end{objective}

\subsection{Contraction Analysis for Nonlinear Systems}\label{ssec:contraction_background}
In this subsection, we present results on nonlinear contraction analysis in \cite{lohmiller1998contraction}, specifically results on discrete-time systems from Section 5 of \cite{lohmiller1998contraction}. This analysis will be used to obtain a fundamental convergence result for the original DeGroot-Friedkin model. The analysis framework that we build will enable an extension to the study of dynamic $\mat{C}$.

Consider a deterministic discrete-time system of the form
\begin{equation}\label{eq:general_discrete_system}
\vect{x}(k+1) = \vect{f}_k(\vect{x}(k), k)
\end{equation}
with $n\times 1$ state vector $\vect{x}$ and $n\times 1$ vector-valued function $\vect{f}$. It is assumed that $\vect{f}$ is smooth, by which we mean that any required derivative or partial derivative exists, and is continuous. The associated virtual\footnote{The term ``virtual'' is taken from \cite{lohmiller1998contraction}; $\delta \vect{x}$ is a virtual, i.e. infinitesimal, displacement.} dynamics is 
\begin{equation*}
\delta\vect{x}(k+1) = \frac{\partial \vect{f}_k}{\partial \vect{x}(k)}\delta{\vect{x}}(k)
\end{equation*}
Define the transformation
\begin{equation*}
\delta{\vect{z}}(k) = \mat{\Theta}_k(\vect{x}(k), k) \delta{\vect{x}}(k)
\end{equation*}
where $\mat{\Theta}_k(\vect{x}(k), k) \in \mathbb{R}^{n\times n}$ is uniformly nonsingular. More specifically, uniform nonsingularity means that there exist a real number $\kappa > 0$ and a matrix norm $\Vert \cdot \Vert^\prime$ such that $\kappa < \Vert \mat{\Theta}_k(\vect{x}(k), k) \Vert^\prime < \kappa^{-1}$ holds for all $\vect{x}$ and $k$. If the uniformly nonsingular condition holds, then exponential convergence of $\delta\vect{z}$ to $\vect{0}_n$ implies, and is implied by, exponential convergence of $\delta\vect{x}$ to $\vect{0}_n$. The transformed virtual dynamics can be computed as
\begin{equation}\label{eq:contract_gen_virt_dyn}
\delta{\vect{z}}(k+1) = \mat{F}(k) \delta\vect{z}(k)
\end{equation}
where $\mat{F}(k) =  \mat{\Theta}_{k+1}(\vect{x}(k+1), k+1) \frac{\partial \vect{f}_k}{\partial \vect{x}(k)} \mat{\Theta}_k(\vect{x}(k), k)^{-1}$ is the transformed Jacobian.

\begin{definition}[Generalised Contraction Region]\label{def:gen_contract_region}
Given the discrete-time system \eqref{eq:general_discrete_system}, a region of the state space is called a generalised contraction region with respect to the metric $\Vert \vect{x} \Vert_{\mat{\Theta}, 1} = \Vert \mat{\Theta}_k(\vect{x}(k), k)\vect{x}(k) \Vert_1$ if in that region, $\Vert \mat{F}(k) \Vert_1 < 1 - \eta$ holds for all $k$, where $\eta > 0$ is an arbitrarily small constant.
\end{definition}

Note that here we are in fact working with the $1$-norm metric in the variable space $\delta{\vect{z}}$ which in turn leads to a \emph{weighted} $1$-norm in the variable space $\delta{\vect{x}}$. Here, the weighting matrix is $\mat{\Theta}_k(\vect{x}(k),k)$ and the weighted $1$-norm is well defined over the entire state space because $\mat{\Theta}$ is required to be uniformly nonsingular.

\begin{theorem}\label{thm:gen_contraction_result}
Given the system \eqref{eq:general_discrete_system}, consider a tube of constant radius with respect to the metric $\Vert \vect{x} \Vert_{\mat{\Theta}, 1}$, centred at a given trajectory of \eqref{eq:general_discrete_system}. Any trajectory, which starts in this tube and is contained at all times in a generalised contraction region, remains in that tube and converges exponentially fast to the given trajectory as $k \to \infty$.

Furthermore, global exponential convergence to the given trajectory is guaranteed if the whole state space is a generalised contraction region with respect to the metric $\Vert \vect{x} \Vert_{\mat{\Theta}, 1}$.
\end{theorem}

Detailed proof of the theorem can be found in the seminal paper \cite{lohmiller1998contraction}, but with a focus on contraction in the Euclidean metric $\Vert \vect{x} \Vert_{\mat{\Theta}, 2} = \Vert \mat{\Theta}_k(\vect{x}(k), k)\vect{x}(k) \Vert_2$, as opposed to the absolute sum metric. However, norms other than the Euclidean norm can be studied because the solutions of \eqref{eq:contract_gen_virt_dyn} can be \emph{superimposed}. This is because \eqref{eq:contract_gen_virt_dyn} around a specific trajectory $\vect{x}(k)$ represents a linear time-varying system in $\delta{\vect{z}}$ coordinates (Section 3.7,  \cite{lohmiller1998contraction}). In the paper, we require use of the $1$-norm metric because the $2$-norm metric does not deliver a convergence result. We provide a sketch of the proof here, modified for the $1$-norm metric, and refer the reader to \cite{lohmiller1998contraction} for precise details.

\begin{proof}
In a generalised contraction region, there holds
\begin{align*}
\Vert \delta{\vect{z}}(k+1) \Vert_1 & = \Vert \mat{F}(k) \delta{\vect{z}}(k) \Vert_1 \nonumber \\
\Vert \delta{\vect{z}}(k+1)\Vert_1 & < (1-\eta) \Vert\delta{\vect{z}}(k)\Vert_1
\end{align*}
since $\Vert \mat{F}(k)\Vert_1 < 1 - \eta$ holds for all $k$ inside the generalised contraction region\footnote{We need $\eta > 0$ to eliminate the possibility that $\lim_{k\to \infty} \Vert \mat{F}(k) \Vert_1 = 1$, which would not result in exponential convergence.}. This implies that $\lim_{k\to\infty} \delta{\vect{z}}(k) = \vect{0}_n$ exponentially fast, which in turn implies that $\lim_{k\to\infty} \delta{\vect{x}}(k) = \vect{0}_n$ exponentially fast due to uniform nonsingularity of $\mat{\Theta}_k(\vect{x}(k),k)$. The definition of $\delta{\vect{x}}$ then implies that \emph{any two infinitesimally close trajectories} of \eqref{eq:general_discrete_system} converge to each other exponentially fast.

The distance between two points, $P_1$ and $P_2$, with respect to the metric $\Vert \cdot\Vert_{\mat{\Theta},1}$ is defined as the shortest path length between $P_1$ and $P_2$, i.e., the smallest path integral $\int_{P_1}^{P_2} \Vert \delta{\vect{z}} \Vert_1 = \int_{P_1}^{P_2} \Vert \delta{\vect{x}} \Vert_{\mat{\Theta},1}$. A tube centred about a trajectory $\vect{x}_1(k)$ and with radius $R$ is then defined as the set of all points whose distances to $\vect{x}_1(k)$ with respect to $\Vert \cdot\Vert_{\mat{\Theta},1}$ are strictly less than $R$.

Let $\vect{x}_2(k) \neq \vect{x}_1(k)$ be any trajectory that starts inside this tube, separated from $\vect{x}_1(k)$ by a finite distance with respect to the metric $\Vert \cdot\Vert_{\mat{\Theta},1}$. Suppose that the tube is contained at all times in a generalised contraction region. The fact that $\lim_{k\to\infty} \Vert \delta{\vect{x}}(k) \Vert_{\mat{\Theta},1} = 0$ then implies that $\lim_{k\to\infty} \int_{\vect{x}_1(k)}^{\vect{x}_2(k)} \Vert \delta{\vect{x}} (k) \Vert_{\mat{\Theta},1} = 0$ exponentially fast. That is, given the trajectories $\vect{x}_2(k)$ and $\vect{x}_1(k)$, separated by a finite distance with respect to the metric $\Vert \cdot\Vert_{\mat{\Theta},1}$, $\vect{x}_2(k)$ converges to $\vect{x}_1(k)$ exponentially fast. Global convergence is obtained by setting $R = \infty$.
\end{proof}

\begin{corollary}\label{cor:convex_contract_region}
If the contraction region is convex, then all trajectories converge exponentially fast to a unique trajectory.
\begin{proof}
This immediately follows because any finite distance between two trajectories shrinks exponentially in the convex region.
\end{proof}
\end{corollary}

\section{Contraction Analysis for Constant $\mat{C}$}\label{sec:constant_C}

In this section, before we address dynamic topology in Section~\ref{sec:dyn_top}, we derive a convergence result for the constant DeGroot-Friedkin model \eqref{eq:DF_system} (i.e., $\mat{C}$ is constant for all $s\in \mathcal{S}$) using nonlinear contraction analysis methods as detailed in Section~\ref{ssec:contraction_background}. The framework built using nonlinear contraction analysis is then applied in the next section to the DeGroot-Friedkin Model with dynamic topology.

In order to obtain a convergence result, we make use of two properties of $\vect{F}(\vect x(s))$ established in \cite{jia2015opinion_SIAM}, but it must be noted that beyond these two properties, the analysis method is novel.

\begin{property}\label{prty01:continuousDF}
The map $\vect{F}(\vect{x}(s))$ is continuous on $\Delta_n$.
\end{property}

If $\mathcal{G}$ does not have star topology, then the following \emph{contraction-like} property holds [pp. 390, Appendix F, \cite{jia2015opinion_SIAM}].
\begin{property}\label{prty02:contraction_like}
Define the set $\mathcal{A} = \{\vect{x} \in \Delta_n : 1-r \geq x_i \geq 0, \forall\,i \in \{1, \hdots, n\}\}$, where $r \ll 1$ is a small strictly positive scalar. Then, there exists a sufficiently small $r$ such that $x_i(s) \leq 1-r$ implies $x_i(s+1) < 1-r$, for all $i$.
\end{property}
By choosing $r$ sufficiently small, it follows that $\vect{x}(s) \in \mathcal{A}, \forall\, s > 0$. In other words, $\vect{F}(\mathcal{A}) \subset \mathcal{A}$. We term this a \emph{contraction-like} property so as not to confuse the reader with our main result; this property establishes a contraction only near the boundary of the simplex $\Delta_n$.

As a consequence of the above two properties, one can easily show, using Brouwer's Fixed Point Theorem (as shown in \cite{jia2015opinion_SIAM}), that there exists \emph{at least} one fixed point $\vect{x}^* = \vect{F}(\vect{x}^*)$ in the convex compact set $\mathcal{A}$. In \cite{jia2015opinion_SIAM}, a method involving multiple inequalities is used to show that the fixed point $\vect{x}^*$ is unique. This is done separately to the convergence proof. In the following proof, we are able to establish exponential convergence to a fixed point, and as a consequence of the method used, immediately prove that it is unique. Lastly, we present a third, easily verifiable property.

\begin{property}\label{prty03:alpha_positive}
If $\vect{x}(s_1) \in \wt{\Delta}_n$ for some $s_1 < \infty$, then $\vect{x}(s) \in \text{int}(\Delta)_n$ for all $s > s_1$.
\end{property}
\begin{proof}
Since $\vect{x}(s_1) \in \wt{\Delta}_n$, $\exists\,j : x_j(s_1) > 0$. In addition, $\gamma_i > 0,\,\forall\,i$ because $\mat{C}$ is irreducible. It then follows that $\alpha(\vect{x}(s_1)) > 0$, and thus $x_i(s_1+1) > 0,\forall\,i$. Thus, $\vect{x}(s) \in \text{int}(\Delta_n)$ for all $s > s_1$.
\end{proof}

% \begin{lem}[Corollary 9.3.4 \cite{bernstein2009matrixbook}]\label{lem:spec_norm}
% Suppose that $\mat{A} \in \mathbb{R}^{n\times n}$ and $\rho(\mat{A}) = \max_i\{|\lambda_{i}(\mat{A})|\} < 1$. Then, there exists a submultiplicative norm $\Vert \cdot \Vert_\alpha$ on $\mathbb{R}^{n\times n}$ such that $\Vert \mat{A}\Vert_\alpha < 1$.
% \end{lem}

% \begin{lem}[Proposition 9.3.1 \cite{bernstein2009matrixbook}]\label{lem:comp_norm}
% Let $\Vert \cdot \Vert_\alpha$ be a submultiplicative norm on $\mathbb{R}^{n \times n}$, and let $\vect{y} \in \mathbb{R}^n$ be nonzero. Then, $\Vert \vect{x} \Vert_\beta \triangleq \Vert \vect{x} \vect{y}^\top \Vert_\alpha$ is a norm on $\mathbb{R}^n$ and $\Vert \cdot \Vert_\beta$ is compatible with $\Vert \cdot \Vert_\alpha$.
% \end{lem}

\subsection{Fundamental Contraction Analysis}

We now state a fundamental convergence result of the system \eqref{eq:DF_system}. In the original work \cite{jia2015opinion_SIAM}, LaSalle's Invariance Principle for discrete-time systems was used to prove an asymptotic convergence result. The result in this paper strengthens this by establishing exponential convergence. In the following proof, when we say a property holds uniformly, we mean that the property holds for all $\vect{x}(s) \in \mathcal{A}$.

\begin{theorem}\label{thm:contract_DF}
Suppose that $n\geq 3$ and suppose further that $\mat{C}$ satisfies Assumption~\ref{assm:C_matrix} and the associated $\mathcal{G}$ does not have star topology. The system \eqref{eq:DF_system}, with initial conditions $\vect{x}(0) \in \wt{\Delta}_n$, converges exponentially fast to a unique equilibrium point $\vect{x}^* \in \text{int}(\Delta_n)$.
\begin{proof}
Consider any given initial condition $\vect{x}(0) \in \wt{\Delta}_n$. According to Property~\ref{prty02:contraction_like}, $\vect{x}(s) \in \mathcal{A}, \forall\, s > 0$ for a sufficiently small $r$. It remains for us to study the system \eqref{eq:DF_system} for $\vect{x}(s) \in \mathcal{A}$. Therefore, in the following analysis, we assume that $s > 0$. The proof heavily utilises the concepts and terminology of Section~\ref{ssec:contraction_background}. 
%The reader is encouraged to refer to it often.

Define the Jacobian of $\vect{F}(\vect{x}(s))$ at the $s^{th}$ issue as $\mat{J}_{\vect{F}}(\vect{x}(s)) = \{\frac{\partial F_i}{\partial x_j} (\vect{x}(s))\}$. We obtain, for $j = i$,
\begin{align}
\frac{\partial F_i}{\partial x_i}(\vect{x}(s)) & = \frac{\gamma_i \alpha(\vect x(s)) }{(1 - x_i(s))^2} - \frac{\gamma_i^2 \alpha(\vect x(s))^2 }{(1-x_i(s))^3} \nonumber \\
& = x_i(s+1) \frac{1 - x_i(s+1)}{1 - x_i(s)} \label{eq:J_diag}
\end{align}
Similarly, we obtain, for $j \neq i$, 
\begin{align}
\frac{\partial F_i}{\partial x_j}(\vect{x}(s)) & = - \frac{\gamma_i \gamma_j \alpha(\vect x(s))^2 }{(1-x_i(s))(1 - x_j(s))^2} \nonumber \\
& = - \frac{x_i(s+1) x_j(s+1) }{1 - x_j(s)}\label{eq:J_offdiag}
\end{align}
Accordingly, we have the following virtual dynamics 
\begin{equation*}
\delta \vect{x}(s+1) = \mat{J}_{\vect{F}}(\vect{x}(s)) \delta \vect{x}(s)
\end{equation*}
Note that $\mat{J}_{\vect{F}}(\vect{x}(s))$ is uniformly well defined and continuous because $x_i(s) < 1-r, \forall\, i,s$, thus enabling nonlinear contraction analysis to be used.

Because there are scenarios where $|\lambda_{\max}(\mat{J}_{\vect{F}}(\vect{x}(s)))| > 1$ (as observed in our simulations), this implies that it is not always possible to find a matrix norm such that $\Vert \mat{J}_{\vect{F}}(\vect{x}(s)) \Vert < 1$ uniformly. We are therefore motivated to seek a contraction result via a coordinate transform. However, rather than study a transformation of $\vect{x}(s)$, we will study a transformation of the virtual displacement $\delta \vect{x}(s)$ as detailed in Section~\ref{ssec:contraction_background}. Specifically, consider the following transformed virtual displacement
\begin{equation}\label{eq:gen_virt_disp}
\delta \vect{z}(s) = \mat{\Theta}(\vect{x}(s),s) \delta \vect{x}(s)
\end{equation}
where $\mat{\Theta}(\vect{x}(s),s) = diag[1/(1-x_i(s))]$, i.e., $\mat{\Theta}$ is a diagonal matrix with the $i^{th}$ diagonal element being $1/(1-x_i(s))$.
% \begin{equation}\label{eq:theta}
% \mat{\Theta}(\vect{x}(s),s) = 
% \begin{bmatrix}
% \frac{1}{1- x_1(s)} & & \\
% & \ddots & \\
% & & \frac{1}{1-x_n(s)}
% \end{bmatrix}
% \end{equation}
It should be noted here that $\mat{\Theta}(\vect{x}(s),s)$ in this proof explicitly depends only on the argument $\vect{x}(s)$, unlike the general result presented in Section~\ref{ssec:contraction_background}, and so we shall write it henceforth as $\mat{\Theta}(\vect{x}(s))$.

The contraction-like Property~\ref{prty02:contraction_like} establishes that $1 > 1-x_i(s) > r > 0$, which in turn implies that $\mat{\Theta}(\vect{x}(s))$ is uniformly nonsingular, with $\lambda_{\min}\big(\mat{\Theta}(\vect{x}(s))\big) > 1$ and $\lambda_{\max}\big(\mat{\Theta}(\vect{x}(s))\big) < 1/r$. In other words, $\kappa < \Vert \mat{\Theta}(\vect{x}(s)) \Vert_1 < \kappa^{-1}$ for some $\kappa > 0$, $\forall\,\vect{x}(s) \in \mathcal{A}$, as required in Section~\ref{ssec:contraction_background}. 

The transformed virtual dynamics is given by
\begin{align}
\delta \vect{z}(s+1) & = \mat{\Theta}(\vect{x}(s+1)) \mat{J}_{\vect{F}}(\vect{x}(s)) \mat{\Theta}(\vect{x}(s))^{-1} \delta \vect{z}(s) \nonumber\\
& = \bar{\mat{H}}(\vect{x}(s)) \delta\vect{z}(s) \label{eq:gen_virt_dyn}
\end{align}
where $\bar{\mat{H}}(\vect{x}(s)) = \mat{\Theta}(\vect{F}(\vect{x}(s))) \mat{J}_{\vect{F}}(\vect{x}(s)) \mat{\Theta}(\vect{x}(s))^{-1}$ is the Jacobian associated with the transformed virtual dynamics. By denoting $\bar{\mat{\Phi}}(\vect{x}(s)) = \mat{J}_{\vect{F}}(\vect{x}(s)) \mat{\Theta}(\vect{x}(s))^{-1}$, one can write $\bar{\mat{H}}(\vect{x}(s)) = \mat{\Theta}(\vect{F}(\vect{x}(s))) \bar{\mat{\Phi}}(\vect{x}(s))$.

The matrix $\bar{\mat{\Phi}}(\vect{x}(s))$ is computed in \eqref{eq:phi_calc} below, and note that it can be considered as being solely dependent on $\vect{x}(s+1) = \vect{F}(\vect{x}(s))$. Therefore, we let $\mat{\Phi}(\vect{x}(s+1)) = \bar{\mat{\Phi}}(\vect{x}(s))$. For brevity, we drop the argument $\vect{x}(s+1)$ where there is no ambiguity and write simply $\mat{\Phi}$.
\begin{figure*}
\begin{align}\label{eq:phi_calc}
\bar{\mat{\Phi}}(\vect{x}(s)) & = 
\begin{bmatrix}
x_1(s+1)\frac{1-x_1(s+1)}{1-x_1(s)} & -\frac{x_1(s+1) x_2(s+1)}{1-x_2(s)} & \cdots & - \frac{x_1(s+1) x_n(s+1)}{1-x_n(s)} \\
-\frac{x_1(s+1) x_2(s+1)}{1-x_1(s)} & x_2(s+1)\frac{1-x_2(s+1)}{1-x_2(s)} & \hdots & \vdots \\ 
\vdots & \vdots & \ddots & \vdots \\
-\frac{x_1(s+1) x_n(s+1)}{1-x_1(s)} & -\frac{x_2(s+1) x_n(s+1)}{1-x_2(s)} & \hdots & \quad x_n(s+1)\frac{1-x_n(s+1)}{1-x_n(s)}\\
\end{bmatrix}\times 
\begin{bmatrix}
1- x_1(s)& & \\
& \ddots & \\
& & 1-x_n(s)
\end{bmatrix} \nonumber \\
& = 
\begin{bmatrix}
x_1(s+1)\big(1-x_1(s+1)\big) & -x_1(s+1) x_2(s+1) & \hdots & -x_1(s+1) x_n(s+1) \\
-x_1(s+1) x_2(s+1) & x_2(s+1)\big(1-x_2(s+1)\big) & \hdots & \vdots \\
\vdots & \vdots & \ddots & \vdots \\
-x_1(s+1) x_n(s+1) &  -x_2(s+1) x_n(s+1) & \hdots & \quad x_n(s+1)\big(1-x_n(s+1)\big)
\end{bmatrix}
\end{align}
\end{figure*}
Note that for each row $i$, $\phi_{ii} = x_i(s+1) \big(1-x_i(s+1)\big)$ and \mbox{$\phi_{ij} = - x_i(s+1) x_j(s+1)$} where $\phi_{ij}$ is the $(i,j)^{th}$ element of $\mat{\Phi}$. From the fact that $0 < x_i(s) < 1-r, \forall\, i$, it follows that all diagonal entries of $\mat{\Phi}$ are uniformly strictly positive and all off-diagonal entries of $\mat{\Phi}$ are uniformly strictly negative. Notice that $\mat{\Phi} = \mat{\Phi}^\top$. Lastly, for any row $i$, there holds
\begin{align*}
\sum_{j = 1}^n \phi_{ij} & = x_i(s+1) \big[ 1 - x_i(s+1) - \sum_{j = 1, j\neq i}^n x_j(s+1) \big] = 0
\end{align*}
because $x_i(s+1) + \sum_{j = 1, j\neq i}^n x_j(s+1) = 1$. In other words, $\mat{\Phi}$ has row and column sums equal to $0$. We thus conclude that $\mat{\Phi}$ is the weighted Laplacian associated with an undirected, \emph{completely connected}\footnote{By completely connected, we mean that there is an edge going from every node $i$ to every other node $j$.} graph with edge weights which vary with $\vect{x}(s+1)$. The edge weights, $-\phi_{ij}$, are uniformly lower bounded away from zero and upper bounded away from 1. This implies that $0 = \lambda_{1}(\mat{\Phi}) < \lambda_{2}(\mat{\Phi}) \leq \hdots \leq \lambda_{n}(\mat{\Phi}) < \infty$ \cite{godsil2001algebraic}, i.e., $\mat{\Phi}$ is uniformly positive semidefinite with a single eigenvalue at $0$, with the associated eigenvector $\vect{1}_n$.

Since $\bar{\mat{\Phi}}(\vect{x}(s)) = \mat{\Phi}(\vect{x}(s+1))$ and  $\mat{\Theta}(\vect{x}(s+1)) = \mat{\Theta}(\vect{F}(\vect{x}(s)))$, we note that $\bar{\mat{H}}(\vect{x}(s))$ can be considered as depending solely on $\vect{x}(s+1)$. Letting $\mat{H}(\vect{x}(s+1)) = \bar{\mat{H}}(\vect{x}(s))$, we complete the calculation $\mat{H}(\vect{x}(s+1)) = \mat{\Theta}(\vect{x}(s+1)) \mat{\Phi}(\vect{x}(s+1))$ to obtain that, for any $i\in \{1, \ldots, n\}$, 
\begin{align*}
{h}_{ii}(\vect{x}(s+1)) & = x_i(s+1) \\
{h}_{ij}(\vect{x}(s+1)) & = -\frac{x_i(s+1) x_j(s+1)}{1-x_i(s+1)}\, , \quad j \neq i
\end{align*}
where ${h}_{ij}(\vect{x}(s+1))$ is the $(i,j)^{th}$ element of $\mat{H}(\vect{x}(s+1))$. For brevity, and when there is no risk of ambiguity, we drop the argument $\vect{x}(s+1)$ and simply write $\mat{H}$. We note that the diagonal entries and off-diagonal entries of $\mat{H}(\vect{x}(s+1))$ are uniformly strictly positive and uniformly strictly negative, respectively. Notice that $\mat{\Phi} \vect{1}_n = \vect{0}_n \Rightarrow \mat{H}\vect{1}_n = \mat{\Theta}(\vect{x}(s+1)) \mat{\Phi}(\vect{x}(s+1)) \vect{1}_n = \vect {0}_n$. In other words, each row of $\mat{H}$ sums to zero. It follows that $\mat{H}$ is the weighted Laplacian matrix associated with a directed, \emph{completely connected} graph with edge weights which vary with $\vect{x}(s+1)$. The edge weights, $-h_{ij}$, are uniformly upper bounded away from infinity and lower bounded away from zero. It is well known that if a directed graph contains a directed spanning tree, the associated Laplacian matrix has a single eigenvalue at $0$, and all other eigenvalues have positive real parts \cite{ren2005consensus}. 

With $\mat{A} = \mat{\Theta}(\vect{x}(s+1))$ uniformly positive definite and $\mat{B} = \mat{\Phi}(\vect{x}(s+1))$ uniformly positive semidefinite, it follows from Lemma~\ref{cor:AB_real} that $\mat{H} = \mat{A}\mat{B}$ has a single zero eigenvalue and \emph{all other eigenvalues are strictly positive and real}. By observing that $\text{trace}(\mat{H}) = \sum_{i=1}^n x_i(s+1) = 1 = \sum_{i=1}^n \lambda_i(\mat{H})$, we conclude that $\max_i \big(\lambda_{i}(\mat{H})\big) < 1$ uniformly, since $n\geq 3$. 

We now establish the stronger result that $\Vert \mat{H} \Vert_1 < 1$ uniformly, which is required to obtain our stability result. See Remark~\ref{rem:schur_stable} below for more insight. Observe that $\Vert \mat{H} \Vert_1 < 1$ if and only if, for all $i \in \{1, \hdots, n\}$, there holds $ \sum_{j=1}^n \vert h_{ji} \vert < 1$, or equivalently,
\begin{align}
x_i + \sum_{j= 1, j \neq i}^n \left( \frac{x_i}{1-x_j} \right) x_j & < 1 \label{eq:H_norm_ineq}
\end{align}
and notice that we have dropped the time argument $s+1$ for brevity. From the fact that $x_i > 0, \forall\, i$ (recall $\alpha(\vect{x}(s)) > 0$), and $n \geq 3$, we obtain $x_i + x_j < 1 \Rightarrow x_i/(1-x_j) < 1$ for all $j \neq i$. Combining this with the fact that $x_i + \sum_{j=1, j\neq i}^n x_j = 1$, we immediately verify that \eqref{eq:H_norm_ineq} holds for all $i$.  Because $\mathcal{A}$ is bounded, this implies that $\Vert \mat{H} \Vert_1 < 1 - \eta$ for some $\eta > 0$ and all $\vect{x}(s) \in \mathcal{A}$. Recalling the transformed virtual dynamics in \eqref{eq:gen_virt_dyn}, we conclude that
\begin{align*}
\Vert \delta \vect{z}(s+1) \Vert_1 & = \Vert \mat{H}(\vect{x}(s+1))\delta \vect{z}(s) \Vert_1 < (1-\eta) \Vert \delta\vect{z}(s) \Vert_1
\end{align*}
We thus conclude that the transformed virtual displacement $\delta \vect{z}$ converges to zero exponentially fast. Recall the definition of $\delta\vect{z}(s)$ in \eqref{eq:gen_virt_disp}, and the fact that $\mat{\Theta}(\vect{x}(s))$ is uniformly nonsingular. It then follows that $\delta\vect{x}(s) \to \vect{0}_n$ exponentially, $\forall\,\vect{x}(s) \in \mathcal{A}$. 

We have thus established that $\mathcal{A}$ is a \emph{generalised contraction region} in accordance with Definition~\ref{def:gen_contract_region}. Because $\mathcal{A}$ is compact and convex, we conclude from Theorem~\ref{thm:gen_contraction_result} and Corollary~\ref{cor:convex_contract_region} that \emph{all trajectories} of $\vect{x}(s+1) = \vect{F}(\vect{x}(s))$ with $\vect{x}(0) \in \wt{\Delta}_n$, converge exponentially \emph{to a single trajectory}. According to Brouwer's Fixed Point Theorem, there is at least one fixed point $\vect{x}^* = \vect{F}(\vect{x}^*) \in \text{int}(\Delta_n)$, \emph{which is a trajectory of $\vect{x}(s+1) = \vect{F}(\vect{x}(s))$}. It then immediately follows that all trajectories of $\vect{x}(s+1) = \vect{F}(\vect{x}(s))$ converge exponentially to \emph{a unique fixed point} $\vect{x}^*\in \text{int}(\Delta_n)$ (recall Property~\ref{prty03:alpha_positive}). 
\end{proof}
\end{theorem}

\begin{corollary}[Vertex Equilibrium]
The fixed point $\mathbf{e}_i$ of the map $\vect{F}(\vect{x})$ is unstable if $\gamma_i < 1/2$. If $\gamma_i = 1/2$, i.e., $v_i$ is the centre node of a star graph, then the fixed point $\mathbf{e}_i$ is asymptotically stable, but is not exponentially stable.
\end{corollary}
\begin{proof}
Without loss of generality, consider $\mathbf{e}_1$. One can avoid $\vect{F}(\vect{x})$ in \eqref{eq:map_F_DF} (and its Jacobian) misbehaving as $\vect{x} \to \mathbf{e}_1$ by multiplying $\alpha(\vect{x})$ by $1/(1-x_1)$ and by multiplying each entry $\gamma_i/(1-x_i)$ by $1-x_1$. One can then differentiate and obtain $\vect{J}_{\vect{F}(\vect{x})}$ and evaluate it at $\vect{x} = \mathbf{e}_1$. Specifically, we obtain $\partial F_1/\partial x_1 = (1-\gamma_1)/\gamma_1$, $\partial F_i/\partial x_1 = -\gamma_i/\gamma_1$, $\partial F_i/\partial x_j = 0$ for all $i,j \neq 1$. Note that this immediately proves that $\vect{F}(\vect{x})$ is continuous at each vertex of the simplex $\Delta_n$, greatly simplifying the proof in Lemma 2.2 of \cite{jia2015opinion_SIAM}.

It follows that $\mat{J}_{\vect{F}(\vect{x})}$ has a single eigenvalue at $(1-\gamma_1)/\gamma_1$ and all other eigenvalues are $0$. If $\gamma_1 < 1/2$, then $(1-\gamma_1)/\gamma_1 > 1$ and the fixed point $\mathbf{e}_1$ is unstable. If $\gamma_1 = 1/2$, then $\mat{J}_{\vect{F}(\vect{x})}$ has a single eigenvalue at $1$. A discrete-time counterpart to Theorem 4.15 in \cite{khalil2002nonlinear} (converse Lyapunov theorem) then rules out $\vect{e}_1$ as an \emph{exponentially stable} fixed point of $\vect{F}(\vect{x})$ (asymptotic stability was established in Lemma~\ref{lem:star}). We omit the proof of the discrete-time counterpart to Theorem 4.15 of \cite{khalil2002nonlinear} due to space limitations. 
\end{proof}

\begin{remark}\label{rem:schur_stable}
When we first analyse $\mat{H}$, we establish that $\forall\,i$, $\lambda_i(\mat{H})$ is real, nonnegative and less than 1. This tells us that the trajectories of \eqref{eq:DF_system} about $\vect{x}^*$ are not oscillatory in nature. It also follows that the spectral radius of $\mat{H}$, given by $\rho(\mat{H})$, is strictly less than 1. In other words, $\mat{H}$ is Schur stable, and according to \cite{horn2012matrixbook}, there exists a submultiplicative matrix norm $\Vert \cdot \Vert^\prime$ such that $\Vert \mat H \Vert^\prime < 1$. However, we must recall that $\mat{H}(\vect{x}(s+1))$ is in fact a nonconstant matrix which changes over the trajectory of the system \eqref{eq:DF_system}. It is not immediately obvious, and in fact is not a consequence of the eigenvalue property, that a single submultiplicative matrix norm $\Vert \cdot \Vert^{\prime\prime}$ exists such that $\Vert \mat{H} \Vert^{\prime\prime} < 1$ \textbf{for all} $\vect{x} \in \mathcal{A}$. Existence of such a norm $\Vert \cdot \Vert^{\prime\prime}$ would establish the desired stability property.

In fact, the system $\delta\vect{z}(s+1) = \mat{H}(\vect{x}(s+1))\delta\vect{z}(s)$, with $\mat{H}\in \mathcal{M}$, $\mathcal{M} = \{\mat{H}(\vect{x}(s+1)) : \vect{x}(s+1) \in \mathcal{A} \}$, can be considered as a discrete-time linear switching system with state $\delta\vect{z}$, and thus under arbitrary switching, the system is stable if and only if the \textbf{joint spectral radius} is less than $1$, that is $\rho(\mathcal{M}) = \lim_{k\to\infty} \max_i \{\Vert \mat{H}_{i_1} \hdots \mat{H}_{i_k} \Vert^{1/k} : \mat{H}_i \in \mathcal{M} \} < 1$ \cite{blondel2000joint_spec_rad}. This is of course a more restrictive condition than simply requiring that $\rho(\mat{H}_i) < 1$. It is known that even when $\mathcal{M}$ is finite, computing the joint spectral radius is NP-hard \cite{tsitsiklis1997joint_spec_rad} and the question ``$\rho(\mathcal{M}) \leq 1$?" is an undecidable problem \cite{blondel2000joint_spec_rad}. The problem is made even more difficult because in this paper, the set $\mathcal{M}$ is not finite. We were therefore motivated to prove the stronger, and nontrivial, result that $\Vert \mat{H} \Vert_1 < 1,\forall\,\vect{x} \in \mathcal{A}$ in order to bypass this issue.
\end{remark}

\begin{remark}\label{rem:comment_deltaz}
For the given definition of $\delta{\vect{z}}$ in \eqref{eq:gen_virt_disp}, we are able to obtain $z_i(s+1) = - \ln (1-x_i(s+1))$ where $z_i$ is the $i^{th}$ element of $\vect{z}(\vect{x}(s))$. However, we did not present the above convergence arguments by firstly defining $\vect{z}(\vect{x}(s))$ and then seeking to study $\vect{z}(s+1) = \vect{G}(\vect{z}(s))$. This is because our proof arose from considering $\vect{x}(s+1) = \vect{F}(\vect{x}(s))$ using the nonlinear contraction ideas developed in \cite{lohmiller1998contraction}, which studied stability via differential concepts. It was through \eqref{eq:gen_virt_disp} that we were able to integrate\footnote{Note that in general, the entries of $\mat{\Theta}$ may have expressions which do not have analytic antiderivatives, and thus an analytic $\vect{z}(\vect{x}(s),s)$ cannot always be found, but $\delta \vect{z}(s)$ can always be defined.} and obtain $z_i = -\ln(1-x_i)$. Moreover, it will be observed in the sequel that by conducting analysis on the transformed Jacobian using nonlinear contraction theory, we are able to straightforwardly deal with dynamic relative interaction matrices. 
\end{remark}

\begin{remark}\label{rem:contraction_metrics}
It should be noted that \cite{lohmiller1998contraction} specifically discusses \emph{contraction in the Euclidean metric} $\Vert \delta\vect{z} \Vert_2 = \Vert \mat{\Theta}\delta\vect{x}\Vert_2$. A contraction region in the Euclidean metric requires $\lambda_{\max}\big(\mat{H}(\vect{x}(s))^\top \mat{H}(\vect{x}(s))\big) < 1$ to hold uniformly. This guarantees that $\delta \vect{z}(s)^\top \delta \vect{z}(s) = \delta\vect{x}(s)^\top \mat{M}(\vect{x}(s),s) \delta\vect{x}(s)$ shrinks to zero exponentially fast, where $\mat{M} = \mat{\Theta}^\top \mat{\Theta}$. However, our simulations showed that $\lambda_{\max}\big(\mat{H}(\vect{x}(s))^\top \mat{H}(\vect{x}(s))\big)$ was frequently and significantly greater than $1$, which indicated that $\delta\vect{z}(s)$ defined in \eqref{eq:gen_virt_disp} is not necessarily contracting in the Euclidean metric. This motivated us to consider contraction of $\delta\vect{z}(s)$ in the absolute sum metric, with appropriate adjustments to the proof presented in Section~\ref{ssec:contraction_background}. Such an approach is alluded to in Section 3.7 of \cite{lohmiller1998contraction}. 
\end{remark}

\subsection{Extending the Contraction-like Analysis}
In this subsection, we provide a result which significantly expands Property~\ref{prty02:contraction_like} by providing an explicit value for $r$ and introduces a stronger \emph{contraction-like result}, which is also applicable to social networks with star topology, unlike Property~\ref{prty02:contraction_like} established in \cite{jia2015opinion_SIAM}.

\begin{lemma}\label{lem:contract}
Suppose that $n\geq 3$, $\vect{x}(0) \in \wt{\Delta}_n$, and $\mathcal{G}$ is strongly connected. Define
\begin{equation}\label{eq:lem_contract_rj}
r_j = \frac{1-2\gamma_j}{1-\gamma_j} 
\end{equation}
where $\gamma_j$ is the $j^{th}$ entry of $\vect{\gamma}^\top$. 
If $\mathcal{G}$ does not have star topology, which implies, from Fact~\ref{fact:star_gamma}, that $r_j > 0$, then for any $0 < r \leq r_j$, there holds
\begin{equation}\label{eq:lem_contract_xF}
x_j \leq 1 - r \Rightarrow F_j(\vect{x}) < 1-r
\end{equation}
where $F_j(\vect{x})$ is the $j^{th}$ entry of $\vect{F}(\vect{x})$. 

If $\mathcal{G}$ has star topology with centre node $j$, which implies $r_j = 0$ in accordance with Fact~\ref{fact:star_gamma}, then $\nexists r > 0 : r \leq r_j$, and thus the contraction-like property in \eqref{eq:lem_contract_xF} does not hold. 
\begin{proof}
It has already been shown that for $\vect{x}(0) \in \wt{\Delta}_n$, there holds $\vect{x}(s) \in \text{int}(\Delta_n)$, i.e., $x_i(s) > 0$ for all $i$ and $s> 0$. Consider then $s > 0$. Suppose that $x_j \leq 1 - r$. Then, with $r \leq r_j$, there holds
\begin{align}
F_j(\vect{x}) & = \alpha(\vect{x})\frac{ \gamma_j}{1 - x_j} \nonumber\\
& = \frac{1}{\frac{\gamma_j}{1-x_j}(1+\frac{\sum_{k\neq j}^n \gamma_k/(1-x_k)}{\gamma_j/(1-x_j)})} \frac{\gamma_j}{1-x_j} \nonumber\\
& = \frac{1}{1+\frac{\sum_{k\neq j}^n \gamma_k/(1-x_k)}{\gamma_j/(1-x_j)}} \leq \frac{1}{    1+\sum_{k\neq j}^n \frac{r}{\gamma_j}\frac{\gamma_k}{(1-x_k)}    } \label{eq:lem_contract_01}
\end{align}
because $r \leq 1 - x_j$. From the fact that $1-x_k < 1$, we obtain $\gamma_k/(1-x_k) > \gamma_k$, which in turn implies that the right hand side of \eqref{eq:lem_contract_01} obeys 
\begin{align}
\frac{1}{    1+\sum_{k\neq j}^n \frac{r}{\gamma_j}\frac{\gamma_k}{(1-x_k)}    } & < \frac{1}{1+\sum_{k\neq j}^n \frac{\gamma_k r}{\gamma_j} } \label{eq:lem_contract_04}  \\
& = \frac{1}{1+\frac{(1-\gamma_j)r}{\gamma_j} } \nonumber\\
& = \frac{\gamma_j}{\gamma_j+(1-\gamma_j)r} \label{eq:lem_contract_02}
\end{align}
with the first equality obtained by noting that $\sum_{k\neq j}^n \gamma_k = 1-\gamma_j$ according to the definition of $\vect{\gamma}$. It follows from \eqref{eq:lem_contract_01} and \eqref{eq:lem_contract_02} that
\begin{align*}
1-r- F_j(\vect{x}) & > 1 - r - \frac{\gamma_j}{\gamma_j+(1-\gamma_j)r} \\
& = \frac{\gamma_j + (1-\gamma_j)r - r\gamma_j - (1-\gamma_j)r^2 - \gamma_j}{\gamma_j +(1-\gamma_j)r} & \\
& = \frac{r(1-2\gamma_j) - r^2(1-\gamma_j)}{\gamma_j +(1-\gamma_j)r} \\ 
& = \frac{ r(1-\gamma_j)\left[\frac{1-2\gamma_j}{1-\gamma_j} - r\right] }{\gamma_j +(1-\gamma_j)r} \label{eq:lem_contract_03}
\end{align*}
Substituting in $r_j$ from \eqref{eq:lem_contract_rj} then yields
\begin{equation}
1-r- F_j(\vect{x}) > \frac{ r(1-\gamma_j)(r_j - r) }{\gamma_j +(1-\gamma_j)r} \geq 0
\end{equation}
because $r_j \geq r$. In other words, $1-r > F_j(\vect{x})$, which completes the proof.
\end{proof}
\end{lemma}
 
This contraction-like result is now used to establish an upper bound on the social power of an individual at equilibrium. We stress here that, it appears that no general result exists for analytical computation of the vector $\vect{x}^*$ given $\vect{\gamma}^\top$. Results exist for some special cases, though, such as for doubly stochastic $\mat{C}$ and for $\mathcal{G}$ with star topology\cite{jia2015opinion_SIAM}. While we do not provide an explicit equality relating $x_i^*$ to $\gamma_i$, we do provide an explicit \emph{inequality}.

\begin{corollary}[Upper bound on $x_i^*$]\label{cor:x_i_upper}
Suppose that $n\geq 3$ and $\vect{x}(0) \in \wt{\Delta}_n$. Suppose further that $\mathcal{G}$ is strongly connected, and is not a star graph. Then, $x_i^* < \gamma_i/(1-\gamma_i)$.
\begin{proof}
Lemma~\ref{lem:contract} establishes that, for any $j\in\{1, \hdots, n\}$, if $x_j \geq 1-r_j$, then the map will always contract in that $F_j(\vect{x}(s)) < x_j$. This is proved as follows. Suppose that $x_j \geq 1-r_j$. Define $r = 1-x_j$, which satisfies $r \leq r_j$ as in Lemma~\ref{lem:contract}. Then, we have $F_j (\vect{x}) < 1 - r = x_j$. It is then straightforward to conclude that the map $\vect{F}(\vect{x})$ continues to contract towards the centre of the simplex $\Delta_n$ until $x_i(s) < 1- r_i, \forall\,i$, where $r_i$ is given by \eqref{eq:lem_contract_rj}. 

Suppose that $x_j^* \geq 1 - r_j = \gamma_j/(1-\gamma_j)$. According to the arguments in the paragraph above, we have $F_j (\vect{x}^*) < 1 - r_j \leq x_j^*$. On the other hand, the definition of $\vect{x}^*$ as a fixed point of $\vect{F}$ implies that $x_j^* = F_j(\vect{x}^*)$, which leads to a contradiction. Therefore, $x_j^* < 1 - r_j = \gamma_j/(1-\gamma_j)$ as claimed. 
\end{proof}
\end{corollary}

Note that this result is separate from the result of Theorem~\ref{thm:contract_DF}, which concluded exponential convergence to a unique fixed point, $\vect{x}^*$. Here, we established an upper bound for the values of the entries of the unique fixed point $\vect{x}^*$, i.e., the social power at equilibrium, given $\vect{\gamma}$.

We mention two specific conclusions following from Corollary~\ref{cor:x_i_upper}. Firstly, suppose that $\mathcal{G}$ has star topology with centre node $v_1$. Then, $\gamma_1 = 0.5$ according to Fact~\ref{fact:star_gamma}, and thus $x_i$ does not contract. This is consistent with the findings in \cite{jia2015opinion_SIAM}, i.e., Lemma~\ref{lem:star}. Secondly, suppose that $\mathcal{G}$ is strongly connected and that $\gamma_i < 1/3,\,\forall\, i\in \{1, \hdots, n\}$. Then, no individual in the social network will have more than half of the total social power at equilibrium, i.e., $x_i^* < 1/2,\,\forall\,i \in \{1, \hdots, n\}$. This second result is relevant as it provides a sufficient condition on the social network topology to ensure that no individual has a dominating presence in the opinion discussion.

\begin{remark}\label{rem:bound_tightness}[Tightness of the Bound]
The tightness of the bound $x_i^* < \gamma_i/(1-\gamma_i)$ increases as $\gamma_k$  decreases $\forall\,k\neq i$. This is in the sense that the ratio $x_i^*(1-\gamma_i)/\gamma_i$ approaches $1$ from below as $\gamma_k$ decreases $\forall\,k\neq i$. We draw this conclusion by noting that in order to obtain \eqref{eq:lem_contract_04}, we make use of the inequality $1 - x_k < 1$. From the fact that $1-x_k$ approaches $1$ as $x_k \to 0$, and because the contraction-like property of Lemma~\ref{lem:contract} holds for $x_k \geq \gamma_k/(1-\gamma_k)$, we conclude that the tightness of the bound $x_i^* < \gamma_i/(1-\gamma_i)$ increases as $\gamma_k$ decreases $\forall\,k \neq i$. If there is a single individual $i$ with $\gamma_i \gg \gamma_k,\forall\,k\neq i$, we are in fact able to accurately estimate $x_i^*$. If $\gamma_i \geq 1/3$, and $n$ is large, then we are able to say, with reasonable confidence, that individual $i$ will hold more than half of the total social power at equilibrium, i.e., $x_i^* \geq 0.5$ is highly likely.
\end{remark}

\subsection{Convergence Rate for a Set of $\mat{C}$ Matrices}

We now present a result on the convergence rate for a constant $\mat{C}$ which is in a subset of all possible $\mat{C}$ matrices. 
\begin{lemma}[Convergence Rate]\label{lem:convergence_rate}
Suppose that $\mat{C}\in \mathcal{L}$, where $\mathcal{L} = \{\mat{C} \in \mathbb{R}^{n\times n} : \gamma_i < 1/3, \forall\,i, \; n \geq 3\}$\footnote{According to Fact~\ref{fact:star_gamma}, $\mathcal{L}$ does not contain any $\mat{C}$ whose associated graph has a star topology.} and $\gamma_i$ is the $i^{th}$ entry of the dominant left eigenvector $\vect{\gamma}^\top$ associated with $\mat{C}$. Then, for the system \eqref{eq:DF_system}, with $\vect{x}(0) \in \wt{\Delta}_n$, there exists a finite $s_1$ such that, for all $s \geq s_1$, there holds $
\Vert \mat{J}_{\vect{F}(\vect{x}(s))} \Vert_1 \leq 2\beta - \epsilon < 1 - \eta$, where $\beta = \max_i \gamma_i/(1-\gamma_i) < 1/2$ and $\epsilon, \eta$ are arbitrarily small positive constants. For $s \geq s_1$, the system \eqref{eq:DF_system} contracts to its unique equilibrium point $\vect{x}^*$ with a convergence rate obeying
\begin{equation*}
\Vert \vect{x}^* - \vect{x}(s+1) \Vert_1 \leq (2\beta - \epsilon) \Vert \vect{x}^* - \vect{x}(s) \Vert_1
\end{equation*}
\end{lemma}
\begin{proof}
From Corollary~\ref{cor:x_i_upper}, we conclude that $x_i^* < \beta_i$ where $\beta_i = \gamma_i/(1-\gamma_i) < 1/2$. Defining $\beta = \max_i \beta_i$, we conclude that $x_i^* \leq \beta - \epsilon_1$ for all $i$, where $\epsilon_1$ is an arbitrarily small positive constant. Note that we already established an exponential convergence result in Theorem~\ref{thm:contract_DF} and an asymptotic result in Lemma~\ref{lem:contract}, but that does not imply that $x_i(s) \leq \beta - \epsilon_1$ for some finite $s$. However, we are able to conclude that there exists a strictly positive $\epsilon$ satisfying  $\epsilon/2 < \epsilon_1$ and $s_1 < \infty$ such that $x_i(s) \leq \beta - \epsilon/2$ for all $s\geq s_1$. 

The Jacobian $\mat{J}_{\vect{F}(\vect{x}(s))}$ has column sum equal to 1. We obtain this fact by observing that, for any $i$,
\begin{align}
\frac{\partial F_i}{\partial x_i}& + \sum_{j = 1, j\neq i}^n \frac{\partial F_j}{\partial x_i} \nonumber \\
& = x_i(s+1) \frac{1 - x_i(s+1)}{1 - x_i(s)} - \sum_{j = 1, j\neq i}^n \frac{x_i(s+1) x_j(s+1) }{1 - x_i(s)} \nonumber \\
& = \frac{x_i(s+1)}{1-x_i(s)}\left[ 1 - x_i(s+1) - \sum_{j = 1, j\neq i}^n x_j(s+1) \right] = 0 \nonumber
\end{align}
because $x_i(s+1) + \sum_{j = 1, j\neq i}^n x_j(s+1) = 1$ by definition. Note also that the diagonal entries of the Jacobian are strictly positive and for $s \geq s_1$, there holds $\partial F_i/\partial x_i \leq \beta - \epsilon/2$, $\forall\,i$. This is because $x_i(1-x_i) \leq (\beta-\epsilon/2)(1-\beta+\epsilon/2)$  for $x_i \leq \beta-\epsilon/2 < 0.5$ and $1/(1-x_i) \leq 1/(1-\beta+\epsilon/2)$. Combining the column sum property and the fact that the off-diagonal entries of the Jacobian are strictly negative, we conclude that for $s \geq s_1$, there holds $\Vert \mat{J}_{\vect{F}(\vect{x}(s))}\Vert_1 = 2 \max_i \partial F_i/\partial x_i \leq 2\beta - \epsilon < 1-\eta$ where $\eta$ is an arbitrarily small positive constant.

The quantity $2\beta - \epsilon$, which is a Lipschitz constant associated with the iteration, upper bounds the $1$-norm of the untransformed Jacobian, and therefore is a lower bound on the convergence rate of the system. 
%In the context of Banach's Fixed-Point Theorem, $2\beta - \epsilon$ is referred to as a Lipschitz constant \cite{khamsi2011fixedpoint_book}. 
In fact, under the special assumption that $\gamma_i < 1/3,\, \forall\,i$, we are able to work directly with the Jacobian $\mat{J}_{\vect{F}}$, as opposed to the transformed Jacobian $\mat{H}$. It is in general much more difficult to compute an upper bound on $\Vert \mat{H} \Vert_1$ using $\vect{\gamma}$ and Corollary~\ref{cor:x_i_upper} when $\exists\,i : \gamma_i \geq 1/3$. 
\end{proof}

Note that $\mathcal{L}$ includes many of the topologies likely to be encountered in social networks. Topologies for which $\gamma_i \geq 1/3$ for some $i$ will have an individual who holds more than half the social power at equilibrium. Such topologies are more reflective of autocracy-like or dictatorship-like networks, as opposed to a group of equal peers discussing their opinions.

\section{Dynamic Relative Interaction Topology}\label{sec:dyn_top}

In this section, we will explore the evolution of individual social power when the relative interaction topology is \emph{issue- or individual-driven}, i.e., $\mat{C}(s)$ is a function of $s$. Motivations for dynamic $\mat{C}(s)$ have been discussed in detail in Sections~\ref{sec:intro} and \ref{sec:background}. This section will establish a theoretical result on the problem of dynamic $\mat{C}(s)$, conjectured and studied extensively with simulations in \cite{friedkin2016tevo_power} but without any proofs. In our earlier work \cite{ye2017DF_IFAC}, we provided analysis on the special case of periodically varying $\mat{C}(s)$, showing the existence of a periodic trajectory. This section provides complete analysis for general switching $\mat{C}(s)$ and extends the periodic result in \cite{ye2017DF_IFAC} as a special case.

Suppose that for a given social network with $n \geq 3$ individuals, there is a finite set $\mathcal{C}$ of $P$ possible relative interaction matrices, defined as $\mathcal{C} = \{\mat{C}_p \in \mathbb{R}^{n\times n} : p \in \mathcal{P}\}$ where $\mathcal{P} = \{1, 2, \ldots, P\}$. We assume that Assumption~\ref{assm:C_matrix} holds for all $\mat{C}_p,\, p \in \mathcal{P}$. For simplicity, we assume that $\nexists\,p$ such that the graph $\mathcal{G}_p$ associated with $\mat{C}_p$ has star topology. Let $\sigma(s) : [0,\infty) \to \mathcal{P}$ be a piecewise constant switching signal, determining the dynamic switching as $\mat{C}(s) = \mat{C}_{\sigma(s)}$. Then, the DeGroot-Friedkin model with dynamic relative interaction matrices is given by
\begin{equation}\label{eq:DF_system_dyn}
\vect x(s+1) = \vect F_{\sigma(s)}(\vect x(s))
\end{equation}
where the nonlinear map $\vect F_p(\vect x(s))$ for $p \in \mathcal{P}$, is defined as
\begin{align}\label{eq:map_F_DF_dyn}
\vect F_p( \vect x(s) ) =   \begin{cases} 
   \mathbf e_i & \hspace*{-6pt} \text{if } \vect x(s) = \mathbf e_i \,\, \text{for any } i \\ \\
   \alpha_p (\vect x(s)) \begin{bmatrix} \frac{\gamma_{p, 1}}{1-x_1(s)} \\ \vdots \\ \frac{\gamma_{p, n}}{1-x_n(s)} \end{bmatrix}       & \text{otherwise }
  \end{cases}
\end{align}
where $\alpha_p(\vect x(s)) = 1/\sum_{i=1}^n \frac{\gamma_{p,i}}{1- x_i(s)}$ and $\gamma_{p,i}$ is the $i^{th}$ entry of the dominant left eigenvector of $\mat{C}_p$, \mbox{$\vect{\gamma}_p = [\gamma_{p,1}, \gamma_{p,2}, \hdots, \gamma_{p,n}]^\top$}. Note that the derivation for \eqref{eq:map_F_DF_dyn} is a straightforward extension of the derivation \eqref{eq:map_F_DF} using Lemma 2.2 in \cite{jia2015opinion_SIAM}, from constant $\mat{C}$ to $\mat{C}(s) = \mat{C}_{\sigma(s)}$. We therefore omit this step. 

\begin{remark}
The system \eqref{eq:DF_system_dyn} is a nonlinear discrete-time switching system, which makes analysis using the usual techniques for switched systems difficult. For arbitrary switching, one might typically seek to find a common Lyapunov function, i.e., one which would establish convergence for any fixed value of $p \in \mathcal{P}$. This, however, appears to be difficult (if not impossible) for \eqref{eq:DF_system_dyn}. In the constant $\mat{C}$ case studied in \cite{jia2015opinion_SIAM}, the convergence result relied on $1)$ a Lyapunov function which was dependent on the unique equilibrium point $\vect{x}^*$, and $2)$ LaSalle's Invariance Principle for discrete-time systems. Both $1)$ and $2)$ are invalid when analysing \eqref{eq:DF_system_dyn}. In the case of $1)$, the system \eqref{eq:DF_system_dyn} does not have a unique equilibrium point $\vect{x}^*$ but rather a unique trajectory $\vect{x}^*(s)$ (as will be made clear in the sequel). In the case of $2)$, LaSalle's Invariance Principle is not applicable to general non-autonomous systems.  
\end{remark}

\subsection{Convergence for Arbitrary Switching}

We now state the main result of this section, the proof of which turns out to be fairly straightforward. This is a consequence of the analysis framework arising from the techniques used in the proof of Theorem~\ref{thm:contract_DF}. Note that in the theorem statement immediately below, a relaxation of the initial conditions is made; we no longer require $\sum_i x_i(0) = 1$. A social interpretation of this is given in Remark~\ref{rem:initial_relax} just following the theorem.

\begin{theorem}\label{thm:contract_dyn_DF}
Suppose that $\nexists\,p$ such that $\mat{C}_p \in \mathcal{C}$ is associated with a star topology graph. Then, system \eqref{eq:DF_system_dyn}, with initial conditions $0\leq x_i(0) < 1,\forall\,i$ and $\exists\,j : x_j(0) > 0$, converges exponentially fast to a unique trajectory $\vect{x}^*(s) \in \text{int}(\Delta_n)$. In other words, each individual $i$ forgets its initial estimate of its own social power, $x_i(0)$, at an exponential rate. For any given $s$, $\vect{x}^*(s+1)$ is determined solely by $\vect{\gamma}_{\sigma(s)}$. If $\vect{x}(0) = \mathbf{e}_i$ for some $i$, then $\vect{x}(s) = \mathbf{e}_i$ for all $s$.
\end{theorem}
\begin{proof}
It is straightforward to conclude that Property~\ref{prty01:continuousDF}, as stated at the beginning of Section~\ref{sec:constant_C}, holds for each map $\vect{F}_p$. With initial conditions $x_i(0) < 1$, the map $\vect{F}_{\sigma(0)}(\vect{x}(s)) \neq \mathbf{e}_i$ for any $i$. We also easily verify that with these initial conditions, the matrix $\mat{W}(0)$ is row-stochastic, irreducible and aperiodic, which implies that the opinions converge for $s =0$ as in the constant $\mat{C}$ case. Because $\mat{C}(0)$ is irreducible, this implies that $\gamma_{\sigma(0),i} > 0$ for all $i$, and we conclude that $\alpha_{\sigma(0)}(\vect{x}(0)) > 0$ because $\exists\,j : x_j(0) > 0$. We thus conclude that $\vect{x}(1) = \vect{F}_{\sigma(0)}(\vect{x}(0)) \succ 0$, i.e., for issue $s = 1$, every individual's social power/self-weight is strictly positive, and the sum of the weights is 1. 

Moreover, because $\mat{C}_p$ is irreducible $\forall\,p$, this implies that for any $p$, there holds $\gamma_{p,i} > 0$ for all $i$. It follows that for $s\geq 1$, $\alpha_{\sigma(s)}(\vect{x}(s)) > 0$, which in turn guarantees that $\vect{x}(s+1) = \vect{F}_{\sigma(s)}(\vect{x}(s)) \succ 0$, i.e., $\vect{x}(s) \in \text{int}(\Delta_n)$ for all $s > 0$. This satisfies the requirements \cite{jia2015opinion_SIAM} on $\vect{x}(s)$ which ensures that $\forall\,s$, $\mat{W}(s)$ is row-stochastic, irreducible, and aperiodic, which implies that opinions converge for every issue. If $\vect{x}(0) = \mathbf{e}_i$ for some $i$, then \eqref{eq:map_F_DF_dyn} leads to the conclusion that $\vect{x}(s) = \mathbf{e}_i$ for all $s$.   

Denote the $i^{th}$ entry of $\vect{F}_p$ by $F_{p,i}$. Regarding Property~\ref{prty02:contraction_like}, stated at the beginning of Section~\ref{sec:constant_C}, for each map $\vect{F}_p$, define the set $\mathcal{A}_p(r_p) = \{\vect{x} \in \Delta_n : 1-r_p \geq x_i \geq 0, \forall\,i \in \{1, \hdots, n\}\}$, where $0 < r_p \ll 1$ is sufficiently small such that $x_i(s) \leq 1-r_p$ for all $i$, which implies that $F_{p,i}(\vect{x}(s)) = x_i(s+1) < 1-r_p$. Define $\bar{\mathcal{A}} = \{\vect{x} \in \Delta_n : 1-\bar{r} \geq x_i \geq 0, \forall\,i \in \{1, \hdots, n\}\}$ where $\bar{r} = \min_p r_p$. Because $\vect{F}_p(\bar{\mathcal{A}}) \subset \bar{\mathcal{A}}$, it follows that $\cup_{p=1}^P \mathcal{A}_p \subset \bar{\mathcal{A}}$, and that for the system \eqref{eq:DF_system_dyn}, for all $s> 0$, $\vect{x}(s) \in \bar{\mathcal{A}}$.

Denoting the Jacobian for the system \eqref{eq:DF_system_dyn} at issue $s$ as $\mat{J}_{\vect{F}_{\sigma(s)}} = \{\frac{\partial F_{\sigma(s),i}}{\partial x_j}\}$, we obtain
\begin{align*}
\frac{\partial F_{\sigma(s),i}}{\partial x_i}(\vect{x}(s)) & = \frac{\gamma_{\sigma(s),i} \alpha_{\sigma(s)}(\vect x(s)) }{(1 - x_i(s))^2} - \frac{\left[\gamma_{\sigma(s),i} \alpha_{\sigma(s)}(\vect x(s))\right]^2 }{(1-x_i(s))^3} \nonumber \\
& = x_i(s+1) \frac{1 - x_i(s+1)}{1 - x_i(s)} %\label{eq:J_diag_dyn}
\end{align*}
Similarly, we obtain, for $j \neq i$,
\begin{align*}
\frac{\partial F_{\sigma(s),i}}{\partial x_j}(\vect{x}(s)) & = - \frac{\gamma_{\sigma(s),i} \gamma_{\sigma(s),j} \left[\alpha_{\sigma(s)}(\vect x(s))\right]^2 }{(1-x_i(s))(1 - x_j(s))^2} \nonumber \\
& = - \frac{x_i(s+1) x_j(s+1) }{1 - x_j(s)}%\label{eq:J_offdiag_dyn}
\end{align*}
Comparing to \eqref{eq:J_diag} and \eqref{eq:J_offdiag}, we note that the Jacobian of the non-autonomous system \eqref{eq:DF_system_dyn} with map \eqref{eq:map_F_DF_dyn} is expressible in the same form as the Jacobian of the original system \eqref{eq:DF_system} with map \eqref{eq:map_F_DF}. More precisely, it can be expressed in a form which is dependent on the trajectory of the system, and not explicitly dependent on $s$. Using the same transformation of $\delta\vect{z}$ given in \eqref{eq:gen_virt_disp} with the same $\mat{\Theta}(\vect{x}(s))$, we obtain the exact same transformed virtual dynamics \eqref{eq:gen_virt_dyn}, expressed as 
\begin{align}
\delta\vect{z}(s+1) = \mat{H}(\vect{x}(s+1))\delta\vect{z}(s)
\end{align}
and it was shown in the proof of Theorem~\ref{thm:contract_DF} that, for some arbitrarily small $\eta > 0$, there holds $\Vert \mat{H} \Vert_1 < 1 - \eta$ for all $\vect{x}(s) \in \bar{\mathcal{A}}$, independent of $p\in \mathcal{P}$. It follows that $\delta\vect{x}(s) \to \vect{0}_n$ exponentially fast for all $\vect{x}(s) \in \bar{\mathcal{A}}$. We thus conclude that $\bar{\mathcal{A}}$ is a generalised contraction region. Again, because $\bar{\mathcal{A}}$ is compact and convex, it follows from Theorem~\ref{thm:gen_contraction_result} and Corollary~\ref{cor:convex_contract_region} that all trajectories of $\vect{x}(s+1) = \vect{F}_{\sigma(s)}(\vect{x}(s))$ converge exponentially \emph{to a single trajectory}, which we denote $\vect{x}^*(s)$. We established earlier that $\vect{x}^*(s) \in \text{int}(\Delta_n)$. 

Exponential convergence to a single unique trajectory can be considered from another point of view as the system \eqref{eq:DF_system_dyn} \emph{forgetting its initial conditions at an exponential rate}. Note also that in one sense, $\mat{F}_{\sigma(s)}$ in \eqref{eq:map_F_DF_dyn} is parametrised by $\vect{\gamma}_{\sigma(s)}$. We conclude from these two points that the unique trajectory $\vect{x}^*(s)$ is such that $\vect{x}^*(s+1)$ depends only on $\vect{\gamma}_{\sigma(s)}$.

Finally, following the same analysis as in [pp.393, \cite{jia2015opinion_SIAM}], one can show that $\lim_{s\to\infty} \vect{\zeta}(s) = \vect{x}^*(s)$ and $\lim_{s\to\infty} \mat{W}(\vect{x}(s)) \!=\! \mat{X}^*(s) \!+\! (\mat{I}_n \!-\! \mat{X}^*(s))\mat{C}(s) \!=\! \mat{W}(\vect{x}^*(s))$. 
\end{proof}

The above result implies that the system \eqref{eq:DF_system_dyn}, with initial conditions satisfying $0\leq x_i(0) < 1,\forall\,i$ and $\exists\,j : x_j(0) > 0$, converges to a unique trajectory $\vect{x}^*(s)$ as $s\to \infty$. For convenience in future discussions and presentation of results, we shall call this the \emph{unique limiting trajectory of \eqref{eq:DF_system_dyn}}. This is a limiting trajectory in the sense that $\lim_{s\to\infty} \vect{x}(s) = \vect{x}^*(s)$.

\begin{remark}[Relaxation of the initial conditions]\label{rem:initial_relax}
Theorem~\ref{thm:contract_dyn_DF} contains a mild relaxation of the initial conditions of the original DeGroot-Friedkin model, and provides a more reasonable interpretation from a social context. One can consider $x_i(0)$ as individual $i$'s estimate of its individual social power (or perceived social power) in the group when the social network is first formed and before discussion begins on issue $s = 0$. The original DeGroot-Friedkin model requires $\vect{x}(0) \in \wt{\Delta}_n$ to avoid an autocratic system (an autocratic system is where $\vect{x}(s) = \mathbf{e}_i$ for some $i$, i.e., an individual holds all the social power). However, this is unrealistic because one cannot expect individuals to have estimates such that $\sum_i x_i(0) = 1$. On the other hand, we do show that the unique limiting trajectory satisfies further, as already commented, $\sum x_i(1)=1$, and then easily $\sum x_i(k)=1,\forall k > 1$ and $\vect{x}^*(s) \in \text{int}(\Delta_n)$, i.e., $x_i^*(s) > 0,\forall\,i$ and $\sum_i x_i^*(s) = 1,\forall\,s$. We therefore show that, as long as no individual $i$ estimates its social power to be autocratic $(x_i(0) = 1)$ and at least one individual estimates its social power to be strictly positive $(\exists\,j : x_j(0) > 0)$, then by sequential discussion of issues, \emph{every individual forgets its initial estimate of its individual social power at an exponential rate.} This occurs \textbf{even for dynamic relative interaction topologies}.
\end{remark}

\subsection{Contraction-Like Property with Arbitrary Switching}
We now extend Lemma~\ref{lem:contract}, Corollary~\ref{cor:x_i_upper} and Lemma~\ref{lem:convergence_rate} to the case of dynamic relative interaction matrices. 

\begin{lemma}\label{lem:contract_dyn}
For the system \eqref{eq:DF_system_dyn}, with initial conditions $0\leq x_i(0) < 1,\forall\,i$ and for at least one $k$, $x_k(0) > 0$, define 
\begin{equation}
\bar{r}_j = \frac{1 - 2\bar{\gamma}_j}{1 - \bar{\gamma}_j},\quad j \in \{1, \ldots, n\}
\end{equation}
where $\bar{\gamma}_j = \max_{p\in \mathcal{P}} \gamma_{p,j}$ and $\gamma_{p,j}$ is the $j^{th}$ entry of $\vect{\gamma}_p$. Then, for any $0 < r \leq \bar{r}_j$ and $p \in \mathcal{P}$, there holds
\begin{equation}
x_j \leq 1 - r \Rightarrow F_{p,j}(\vect{x}) < 1-r
\end{equation}
where $F_{p,j}(\vect{x})$ is the $j^{th}$ entry of $\vect{F}_{p}(\vect{x})$.
\end{lemma}
\begin{proof}
The lemma is proved by straightforwardly checking that, for the given definition of $\bar{r}_j$, the result in Lemma~\ref{lem:contract} holds separately for every map $\vect{F}_p, \;p \in \mathcal{P}$. In other words, for all $i,p$, $x_i(s) \leq 1 - r \Rightarrow F_{p,i}(\vect{x}(s)) < 1 - r\,,\forall\,r\leq \bar{r}_i$.
\end{proof}

\begin{corollary}[Upper bound on $x_i^*(s)$]\label{cor:x_i_upper_dyn}
For the system \eqref{eq:DF_system_dyn}, with initial conditions $0\leq x_i(0) < 1,\forall\,i$ and for at least one $j$, $x_j(0) > 0$, there holds $x_i^*(s) \leq \bar{\gamma}_i/(1-\bar{\gamma}_i),\forall\,s$, where $\bar{\gamma}_j = \max_{p\in \mathcal{P}} \gamma_{p,j}$ and $x_i^*(s)$ is the $i^{th}$ entry of the unique limiting trajectory $\vect{x}^*(s)$.
\end{corollary}
\begin{proof}
The proof is a straightforward extension of the proof of Corollary~\ref{cor:x_i_upper}, and is therefore not included here.
\end{proof}

\begin{lemma}[Convergence Rate for Dynamic Topology]\label{lem:convergence_rate_dyn}
For all $p\in\mathcal{P}$, suppose that $\mat{C}_p\in \mathcal{L}$ where $\mathcal{L} = \{\mat{C}_p \in \mathbb{R}^{n\times n} : \gamma_{p,i} < 1/3, \forall\,i\}$ and $\gamma_{p,i}$ is the $i^{th}$ entry of the dominant left eigenvector $\vect{\gamma}_p$ associated with $\mat{C}_p$. Then, there exists a finite $s_1$ such that, for all $s \geq s_1$, there holds $\Vert \mat{J}_{\vect{F}_{\sigma(s)}(\vect{x}(s))} \Vert_1 \leq 2\bar{\beta} - \epsilon < 1 - \eta$, where $\bar{\beta} = \max_p \max_i \gamma_{p,i}/(1-\gamma_{p,i}) < 1/2$ and $\epsilon, \eta$ are arbitrarily small positive constants. For $s \geq s_1$, the system \eqref{eq:DF_system_dyn} contracts to its unique limiting trajectory $\vect{x}^*(s)$ with a convergence rate obeying 
\begin{equation}
\Vert \vect{x}^*(s) - \vect{x}(s+1) \Vert_1 \leq (2\bar{\beta} - \epsilon) \Vert \vect{x}^*(s) - \vect{x}(s) \Vert_1
\end{equation}
\end{lemma}
\begin{proof}
Again, the proof is a straightforward extension of the proof of Lemma~\ref{lem:convergence_rate}, by recalling from the proof of Theorem~\ref{thm:contract_dyn_DF} that the Jacobian takes on the same form. We thus omit the minor details.
\end{proof}

\begin{remark}[Self-Regulation]\label{rem:self_reg}
The exponential forgetting of initial conditions is a powerful notion. It implies that sequential discussion of topics combined with reflected self-appraisal is a method of ``self-regulation" for social networks, even in the presence of dynamic topology. Consider an individual $i$ who is extremely arrogant, e.g. $x_i(0) = 0.99$. However, individual $i$ is not likeable and others tend to not trust its opinions on any issue, e.g. $c_{ji}(s) \ll 1\,,\forall\,j,s$. Then, $\gamma_i(s) \ll 1$ because $\vect{\gamma}(s)^\top = \vect{\gamma}(s)^\top \mat{C}(s)$ implies $\gamma_i(s) = \sum_{j\neq i} \gamma_j(s) c_{ji}(s)$. Then, according to Corollary~\ref{cor:x_i_upper_dyn}, $x_i^*(s) \ll 1$, and individual $i$ exponentially loses its social power. An interesting future extension would be to expand on the reflected self-appraisal by modelling individual \textbf{personality}. For example, we can consider $x_i(s+1) = \phi_i(\zeta_i(s))$ where $\phi_i(\cdot)$ may capture arrogance or humility. 

We also conclude that, for large $s$, any individual wanting to have an impact on the discussion of topic $s+1$ should focus on ensuring it has a large impact on discussion of the prior topic $s$. This concept can be applied to e.g. \cite{ye2017socialpower_mod}.
\end{remark}

\subsection{Periodically Varying Topology}
In this subsection, we investigate an interesting, special case of issue-dependent topology, that of periodically varying $\mat{C}(s)$ which satisfies Assumption~\ref{assm:C_matrix} for all $s$. Preliminary analysis and results were presented in \cite{ye2017DF_IFAC} without convergence proofs. We now provide a complete analysis by utilising Theorem~\ref{thm:contract_dyn_DF}.

\emph{Motivation for Periodic Variations:} Consider \emph{Example 1} in Section~\ref{ssec:problem_def} of a government cabinet that meets to discuss the issues of defence, economic growth, social security programs and foreign policy. Since these issues are vital to the smooth running of the country, we expect the issues to be discussed \emph{regularly and repeatedly}. Regular meetings on the same set of issues for decision making/governance/management of a country or company then points to periodically varying $\mat{C}(s)$, i.e., social networks with periodic topology.

The system \eqref{eq:DF_system_dyn}, with periodically switching $\mat{C}(s)$, can be described by a switching signal $\sigma(s)$ of the form $\sigma(0) = P$, and for $s\geq 1$, $\sigma(Pq + p) = p$,\footnote{Note that any given $s\in \mathcal{S}$ can be uniquely expressed by a given fixed positive integer $P$, a nonnegative integer $q$, and positive $p\in \mathcal{P}$, as shown.} where $P < \infty$ is the period length, $p \in \mathcal{P} = \{1, 2, \hdots, P\}$ and $q \in \mathbb{Z}_{\geq 0}$ is any nonnegative integer. Note that in general, $\mat{C}_i \neq \mat{C}_j, \forall\,i,j \in \mathcal{P}$ and $i\neq j$. Theorem~\ref{thm:contract_dyn_DF} immediately allows us to conclude that system \eqref{eq:DF_system_dyn} with periodic switching converges exponentially fast to its unique limiting trajectory $\vect{x}^*(s)$. This subsection's key contribution is to use a transformation to obtain additional, useful information on the limiting trajectory.

For simplicity, we shall begin analysis by assuming that $\mathcal{P} = \{1,2\}$, i.e., there are two different $\mat{C}$ matrices, and the switching is of period 2. It will become apparent in the sequel that analysis for $\mathcal{P} = \{1, 2, \hdots, P\}$, with arbitrarily large but finite $P$, is a simple recursive extension on the analysis for $\mathcal{P} = \{1,2\}$. For the two matrices case, we obtain
\begin{equation}\label{eq:periodic_DF_system}
\vect x(s+1) = \begin{cases} 
   \vect F_1(\vect x(s)) &  \text{if } s \text{ is odd} \\
   \vect F_2(\vect x(s)) &  \text{if } s \text{ is even}
  \end{cases}
\end{equation}

We now seek to transform the periodic system into a time-invariant system. Define a new state $\vect y \in \mathbb{R}^{2n}$ (note that this is not the opinion state given in Section~\ref{sssec:degroot}) as
\begin{equation}\label{eq:y_definition}
\vect y(2q) = 
\begin{bmatrix} \vect y_1(2q) \\ \vect y_2(2q) \end{bmatrix} = \begin{bmatrix} \vect x(2q) \\ \vect x(2q+1) \end{bmatrix}
\end{equation}
and study the evolution of $\vect y (2q)$ for $q \in \{0, 1, 2, \hdots\}$. Note that
\begin{align}\label{eq:y_evolution}
\vect y(2(q+1)) & = 
\begin{bmatrix} \vect y_1(2(q+1)) \\ \vect y_2(2(q+1)) \end{bmatrix} = \begin{bmatrix} \vect x(2(q+1)) \\ \vect x(2(q+1)+1) \end{bmatrix}
\end{align} 
In view of the fact that $\vect x(2(q+1)) = \vect F_1(\vect x(2q+1))$ and $\vect x(2(q+1)+1) = \vect F_2(\vect x(2q+2))$ for any $q \in \{0,1,2,\hdots\}$, we obtain
\begin{equation}
\vect y(2(q+1)) = 
\begin{bmatrix} \vect F_1(\vect x(2q+1)) \\ \vect F_2(\vect x(2q+2)) \end{bmatrix}
\end{equation}
Similarly, notice that $\vect x(2q+1) = \vect F_2(\vect x(2q))$ and $\vect x(2q+2) = \vect F_1(\vect x(2q+1))$ for any $q \in \{0,1,2,\hdots\}$. From this, for $q \in \{0,1,2,\hdots\}$, we obtain that
\begin{align}
\vect y(2(q+1)) 
% & = \begin{bmatrix} \vect F_1\Big(\vect F_2(\vect x(2q))\Big) \\ \vect F_2\Big(\vect F_1(\vect x(2q+1))\Big) \end{bmatrix} \\
& = \begin{bmatrix} \vect F_1\Big(\vect F_2(\vect y_1(2q))\Big) \\ \vect F_2\Big(\vect F_1( \vect y_2(2q) )\Big) \end{bmatrix}  = \begin{bmatrix} \vect G_1 (\vect y_1(2q)) \\ \vect G_2 (\vect y_2(2q)) \end{bmatrix} \label{eq:y_update_G}
\end{align}
for the time-invariant nonlinear composition functions $\vect G_1 = \vect{F}_1 \circ \vect{F}_2$ and $\vect G_2 = \vect{F}_2 \circ \vect{F}_1$. We can thus express the periodic system \eqref{eq:periodic_DF_system} as the nonlinear time-invariant system
\begin{equation}\label{eq:TI_DF_system}
\vect y(2q+2) = \bar{\vect G} (\vect y(2q))
\end{equation}
where $\bar{\vect{G}} = [\vect{G}_1^\top, \vect{G}_2^\top]^\top$.

\begin{theorem}\label{thm:periodic_2}
The system \eqref{eq:periodic_DF_system}, with initial conditions $0\leq x_i(0) < 1,\forall\,i$ and $\exists\,j : x_j(0) > 0$, converges exponentially fast to a unique limiting trajectory $\vect{x}^*(s) \in \text{int}(\Delta_n)$. This trajectory is a periodic sequence, which obeys 
\begin{equation}\label{eq:periodic_sequence_2}
\vect x^*(s) = \begin{cases} 
   \vect y_1^* &  \text{if } s \text{ is odd} \\
   \vect y_2^* &  \text{if } s \text{ is even}
  \end{cases}
\end{equation}
where $\vect y_1^* \in \text{int}(\Delta_n)$ and $\vect y_2^* \in \text{int}(\Delta_n)$ are the unique fixed points of $\vect{G}_1$ and $\vect G_2$, respectively.
\end{theorem}
\begin{proof}
As mentioned above, one can immediately apply Theorem~\ref{thm:contract_dyn_DF} to show $\lim_{s\to\infty} \vect{x}(s) = \vect{x}^*(s)$. This proof therefore focuses on using the time-invariant transformation to show that $\vect{x}^*(s)$ has the properties described in the theorem statement.

\emph{Part 1:} In this part, we prove that the map $\vect{G}_i$, $i = 1,2$ has at least one fixed point. Firstly, we proved in Theorem~\ref{thm:contract_dyn_DF} that the system \eqref{eq:DF_system_dyn}, with initial conditions $0\leq x_i(0) < 1,\forall\,i$ and for at least one $j$, $x_j(0) > 0$, will have $\vect{x}(s) \in \text{int}(\Delta_n)$ for all $s > 0$, which implies that $\vect{x}^*(s) \in \text{int}(\Delta_n)$. Let $p \in \{1,2\}$. The fact that $\vect F_p : \Delta_n \rightarrow \Delta_n$ is continuous on $\wt{\Delta}_n$ is straightforward since $\vect{F}_p$ is an analytic function in $\wt{\Delta}_n$. Lemma 2.2 in \cite{jia2015opinion_SIAM} shows that $\vect F_p$ is Lipschitz continuous about $\mathbf{e}_i$ with Lipschitz constant $2\sqrt{2}/\gamma_{i,p}$. It is then straightforward to verify that the composition of two continuous functions, $\vect G_1 = \vect F_1 \circ \vect F_2 : \Delta_n \rightarrow \Delta_n$ is continuous. Similarly, $\vect G_2 = \vect F_2 \circ \vect F_1 : \Delta_n \rightarrow \Delta_n$ is also continuous.

The proof of Theorem~\ref{thm:contract_dyn_DF} also showed that for all $p$, $\vect{F}_p \in \bar{\mathcal{A}}$ where $\bar{\mathcal{A}} = \{\vect{x} \in \Delta_n : 1-\bar{r} \geq x_i \geq 0, \forall\,i \in \{1, \hdots, n\}\}$ and $\bar{r}$ is some small strictly positive constant. For the system \eqref{eq:periodic_DF_system} with $p = 1,2$, it follows that $\vect F_1(\bar{\mathcal{A}}) \subset \bar{\mathcal{A}} \Rightarrow \vect F_2(\vect F_1(\bar{\mathcal{A}})) \subset \bar{\mathcal{A}}$, which implies that $\vect G_1(\bar{\mathcal{A}}) \subset \bar{\mathcal{A}}$. Similarly, $\vect{G}_2(\bar{\mathcal{A}})  \subset \bar{\mathcal{A}}$. Brouwer's fixed-point theorem then implies that there exists at least one fixed point $\vect y_1^* \in \bar{\mathcal{A}}$ such that $\vect y_1^* = \vect G_1(\vect y_1^*)$ (respectively $\vect y_2^* \in \bar{\mathcal{A}}$ such that $\vect y_2^* = \vect G_2(\vect y_2^*)$) because $\vect G_1$ (respectively $\vect{G}_2$) is a continuous function on the compact, convex set $\mathcal{A}$. The arguments in \emph{Part 1} appeared in \cite{ye2017DF_IFAC}, but proofs were omitted due to space limitations.

\emph{Part 2:} In this part, we prove that the unique limiting trajectory of \eqref{eq:periodic_DF_system} obeys \eqref{eq:periodic_sequence_2}. Let $\vect{y}_1^*$ be a fixed point of $\vect{G}_1$. We will show below that $\vect{y}_1^*$ is in fact unique. Observe that $\vect y_1^* = \vect F_2( \vect F_1 (\vect y_1^*) )$. Define $\vect y_2^* = \vect F_1(\vect y_1^*)$. We thus have $\vect y_1^* = \vect F_2(\vect y_2^*)$. Observe that $\vect F_1 (\vect y_1^*)  = \vect F_1 (\vect F_2 (\vect y_2^*))$, which implies that $\vect y_2^* = \vect F_1 (\vect F_2(\vect y_2^*)) = \vect G_2 (\vect y_2^*)$. \emph{In other words, $\vect y_2^*$ is a fixed point of $\vect G_2$ (but at this stage we have not yet proved its uniqueness)}. 

We now prove uniqueness. Theorem~\ref{thm:contract_dyn_DF} allows us to conclude that \emph{all trajectories} of \eqref{eq:periodic_DF_system} converge exponentially fast to a unique limiting trajectory $\vect{x}^*(s) \in \text{int}(\Delta_n)$. It follows, from \eqref{eq:TI_DF_system} and the definition of $\vect{y}(2q)$, that for all $s\geq 0$, \eqref{eq:periodic_sequence_2} \emph{is a trajectory} of the system \eqref{eq:periodic_DF_system}; the critical point here is that \eqref{eq:periodic_sequence_2} \emph{holds for all} $s$. Combining these arguments, it is clear that \eqref{eq:periodic_sequence_2} is precisely \emph{the unique limiting trajectory.}

Lastly, we show that $\vect{y}_1^*$ and $\vect{y}_2^*$ are the unique fixed point of $\vect{G}_1$ and $\vect{G}_2$, respectively. To this end, suppose that, to the contrary, at least one of $\vect{y}_1^*$ and $\vect{y}_2^*$ is not unique. Without loss of generality, suppose in particular that $\vect{y}_1^\prime \!\neq\! \vect{y}_1^*$ is any other fixed point of $\vect{G}_1$. Then, $\vect{y}_2^\prime \!=\! \vect{F}_1(\vect{y}_1^\prime)$ is a fixed point of $\vect{G}_2$, and
\begin{equation}\label{eq:periodic_traj_false}
\vect x(s) = \begin{cases} 
   \vect y_1^\prime &  \text{if } s \text{ is odd} \\
   \vect y_2^\prime &  \text{if } s \text{ is even}
  \end{cases}
\end{equation}
is a trajectory of \eqref{eq:periodic_DF_system} that holds for all $s\geq 0$, and is \emph{different from the trajectory \eqref{eq:periodic_sequence_2} because} $\vect{y}_1^\prime \neq \vect{y}_1^*$. On the other hand, Theorem~\ref{thm:contract_dyn_DF} implies that all trajectories of \eqref{eq:periodic_DF_system} converge exponentially fast \emph{to a unique limiting trajectory}, which is a contradiction. Thus, $\vect{y}_1^*$ and $\vect{y}_2^*$ are the unique fixed point of $\vect{G}_1$ and $\vect{G}_2$, respectively, and \eqref{eq:periodic_DF_system} converges exponentially fast to the unique limiting trajectory \eqref{eq:periodic_sequence_2}.
\end{proof}

We now provide the generalisation to periodically switching topology $\mat{C}(s) = \mat{C}_{\sigma(s)}$, where $\sigma(s)$ is of the form $\sigma(0) = P$, and for $s\geq 1$, $\sigma(Pq + p) = p$. Here, $2 \le P < \infty$, $p \in \mathcal{P} = \{1, 2, \hdots, P\}$ and $q \in \mathbb{Z}_{\geq 0}$. The periodic DeGroot-Friedkin model is described by 
\begin{equation}\label{eq:periodic_DF_system_M_top}
\vect x(s+1) = \begin{cases} 
   \vect F_P(\vect x(s)) &  \text{for } s = 0 \\
   \vect F_p(\vect x(s = Pq + p ))  &  \text{for all } s \geq 1 
  \end{cases}
\end{equation}

A transformation of \eqref{eq:periodic_DF_system_M_top} to a time-invariant system can be achieved by following a procedure similar to the one detailed for the case $p=2$. A new state variable $\vect y \in \mathbb{R}^{Pn}$ is defined as 
\begin{equation}\label{eq:y_definition_M}
\vect y(Pq) = 
\begin{bmatrix} \vect y_1(Pq) \\ \vect y_2(Pq) \\ \vdots \\ \vect y_P(Pq) \end{bmatrix} = \begin{bmatrix} \vect x(Pq) \\ \vect x(P q +1 ) \\ \vdots \\ \vect x(P q + P -1) \end{bmatrix}
\end{equation}
and we study the evolution of $\vect y(Pq)$ for $q \in \{0,1, \hdots\}$. It follows that
% \begin{align}\label{eq:y_evolution_M}
% \vect y(P (q+1)) & = 
% \begin{bmatrix} \vect y_1( P(q+1) ) \\ \vect y_2( P(q+1)) \\ \vdots \\ \vect y_P(P(q+1)) \end{bmatrix}
% = \begin{bmatrix} \vect x( P(q +1) ) \\ \vect x(P(q+1)+1) \\ \vdots \\ \vect x (P(q+1) + P-1) \end{bmatrix}
% \end{align}
\begin{align*}
\vect y_p(P (q+1)) = \vect x( P(q +1) + p - 1 )\,, \; \forall\, p \in \mathcal{P}
\end{align*}
Following the logic in the 2 period case, but with the precise steps omitted, we obtain
\begin{align}
\vect y(P (q+1)) &
%\nonumber \\ 
% & \quad = \begin{bmatrix} \vect F_M ( \vect F_{M-1} (\hdots (\vect F_1 (\vect x (M (s-1) +1) ) ) ) ) \\
% \vect F_1 ( \vect F_{M} (\hdots (\vect F_2 (\vect x (M (s-1) +2) ) ) ) ) \\
% \vdots \\
% \vect F_{M-1} ( \vect F_{M-2} (\hdots (\vect F_M (\vect x (M (s-1) + M) ) ) ) ) \end{bmatrix} \nonumber \\ 
 = \begin{bmatrix} \vect F_{P-1} ( \vect F_{P-2} (\hdots (\vect F_P (\vect y_1 ( Pq ) ) ) ) ) \\
\vect F_P ( \vect F_{P-1} (\hdots (\vect F_1 (\vect y_2 ( Pq ) ) ) ) ) \\
\vdots \\
\vect F_{P-2} ( \vect F_{P-1} (\hdots (\vect F_P (\vect y_{P-1} ( Pq ) ) ) ) ) \end{bmatrix}  \nonumber \\
& = \bar {\vect G}(\vect y(Pq)) \label{eq:transformed_DF_M_top} 
\end{align}
where $\bar{\vect{G}}(\vect y)  = [\vect G_1 (\vect y_1 ), 
\vect G_2 ( \vect y_2 ) 
,\hdots, 
\vect G_P ( \vect y_P)]^\top$.
% \begin{align}
% \bar{\vect{G}}&(\vect y(Pq))  \nonumber \\ & = [\vect G_1 (\vect y_1 ( Pq ) ), 
% \vect G_2 ( \vect y_2 ( Pq ) ) 
% ,\hdots, 
% \vect G_P ( \vect y_P ( Pq ) )]^\top
% \end{align}
This leads to the following generalisation of Theorem~\ref{thm:periodic_2}.

\begin{theorem}\label{thm:periodic_P_gen}
The system \eqref{eq:periodic_DF_system_M_top}, with initial conditions $0\leq x_i(0) < 1,\forall\,i$ and for at least one $j$, $x_j(0) > 0$, converges exponentially fast to a unique limiting trajectory $\vect{x}^*(s) \in \text{int}(\Delta_n)$. This trajectory is a periodic sequence, which for any $q \in \mathbb{Z}_{\geq 0}$, obeys 
\begin{equation}\label{eq:periodic_sequence_M}
\vect x^*(Pq + p - 1) = \vect y_p^*, \; \text{for all } p\in \{1, 2, \ldots, P\} 
\end{equation}
where $\vect y_p^* \in \text{int}(\Delta_n)$ is the unique fixed point of $\vect G_p$.
\end{theorem}
\begin{proof}
The proof is obtained by recursively applying the same techniques used in the proof of Theorem~\ref{thm:periodic_2}. We therefore omit the details.
\end{proof}

Note that Lemmas~\ref{lem:contract_dyn} and \ref{lem:convergence_rate_dyn} and Corollary~\ref{cor:x_i_upper_dyn} are all applicable to the periodic system \eqref{eq:periodic_DF_system_M_top} because \eqref{eq:periodic_DF_system_M_top} is just a special case of the general switching system \eqref{eq:DF_system_dyn}.

\subsection{Convergence to a Single Point}\label{ssec:dyn_unique_point}

We conclude Section~\ref{sec:dyn_top} by showing that if the set $\mathcal{C}$ of possible switching matrices has a special property, then the unique limiting trajectory $\vect{x}^*(s) \in \text{int}(\Delta_n)$ is in fact a stationary point.

Define $\mathcal{K}(\wt{\vect{\gamma}}) = \{\mat{C}_p \in \mathbb{R}^{n\times n} : \vect{\gamma}_p = \wt{\vect{\gamma}}, \forall\,p \in \mathcal{P} = \{1, 2, \hdots, P\} \}$ where $P$ is finite. In other words, $\mathcal{K}(\wt{\vect{\gamma}})$ is a set of $\mat{C}$ matrices which all have the same dominant left eigenvector $\wt{\vect{\gamma}}^\top$. Perhaps the most well-known set is $\mathcal{K}(\vect{1}_n/n)$, i.e., the set of $n\times n$ doubly-stochastic $\mat{C}$ matrices. 

\begin{theorem}
Suppose that $\mat{C}(s) = \mat{C}_{\sigma(s)} \in \mathcal{K}(\wt{\vect{\gamma}})$. Then, the system \eqref{eq:DF_system_dyn}, with initial conditions $0\leq x_i(0) < 1,\forall\,i$ and for at least one $j$, $x_j(0) > 0$, converges exponentially fast to a unique point $\vect{x}^* \in \text{int}(\Delta_n)$. 

There holds $x^*_i < x^*_j$ if and only if $\wt{\gamma}_i < \wt{\gamma}_j$, for any $i,j$, where $\wt{\gamma}_i$ and $x^*_i$ are the $i^{th}$ entry of the dominant left eigenvector $\wt{\vect\gamma}$ and $\vect{x}^*$, respectively. There holds $x^*_i = x^*_j$ if and only if $\wt{\gamma}_i = \wt{\gamma}_j$. 
\end{theorem}
\begin{proof}
The map $\vect{F}_{\sigma(s)}$ is parametrised simply by the vector $\vect{\gamma}_{\sigma(s)}$. Under the stated condition of $\mat{C}(s) = \mat{C}_{\sigma(s)} \in \mathcal{K}(\wt{\vect{\gamma}})$, the map $\vect{F}_{\sigma(s)}$ is time-invariant. The result in Theorem~\ref{thm:contract_DF} is then used to complete the proof. 
\end{proof}

\section{Simulations}\label{sec:sim}
In this section, we provide a short simulation for a network with $6$ individuals to illustrate our key results. The set of topologies is given as $\mathcal{C} = \{\mat{C}_1, \hdots, \mat{C}_5\}$, i.e., $\mathcal{P} = \{1, 2, \hdots, 5\}$. The switching signal $\sigma(s)$ is generated such that for any given $s$, there is equal probability that $\sigma(s) = p,\forall\,p \in\mathcal{P}$. The precise numerical forms of $\mat{C}_p$ given in the appendix.

Figure~\ref{fig:IC_set1} shows the evolution of individual social power over a sequence of issues for the system as described in the above paragraph, initialised from a set of initial conditions, $\wh{\vect{x}}(0)$.  Figure~\ref{fig:IC_set2} shows the system with a different set of initial conditions $\wt{\vect{x}}(0) \neq \wh{\vect{x}}(0)$. Notice that individuals $1,2,3$ have large perceived social power $\wh{x}_i(0) = 0.95$, while individuals $4,5,6$ have $\wh{x}_i(0) = 0$. In the other set of initial conditions, $\wt{x}_i(0)$ is large for $i = 4,6$. Through sequential discussion and reflected self-appraisal, it is clear that the initial conditions are exponentially forgotten and both plots show convergence to the same unique limiting trajectory $\vect{x}^*(s)$ by about $s = 10$. This is shown in Fig.~\ref{fig:IC_compare}, which displays the individual social powers of selected individuals $1,3$ and $6$. The solid lines correspond to initial condition set $\wh{\vect{x}}(0)$ while the dotted lines correspond to initial condition set $\wt{\vect{x}}(0)$. Figure~\ref{fig:IC_compare} shows the exponential convergence of the dotted and solid trajectories. Note that for individual $4$, its social power is always strictly positive, although for several issues, $x_4(s)$ is close to $0$.

For each individual, with $\bar{\gamma}_i = \max_{p\in\mathcal{P}} \gamma_{p,i}$, we computed $\bar\gamma_{1} = 0.4737, \bar{\gamma}_{2} = 0.2371, \bar\gamma_3 = 0.2439, \bar\gamma_4 = 0.2439, \bar\gamma_5 = 0.2439, \bar\gamma_6 = 0.2392$. Note that $\sum_i \bar{\gamma}_i \neq 1$ in general due to the definition of $\bar{\gamma}_i$. According to Corollary~\ref{cor:x_i_upper_dyn}, we have $\vect{x}^*(s) \preceq [0.9, 0.3108, 0.3226, 0.3226, 0.3226,0.3144]$. This is precisely what is shown in Figs.~\ref{fig:IC_set1} and \ref{fig:IC_set2}. Since only $\bar{\gamma}_1 > 1/3$, we observe that after the first 10 or so issues, only $x_1^*(s) > 0.5$, i.e., only individual $1$ can hold more than half the social power in the limit, under arbitrary switching. Simulations for periodically-varying topology are available in \cite{ye2017DF_IFAC}.

\begin{figure}
\begin{center}
\includegraphics[width=0.85\linewidth]{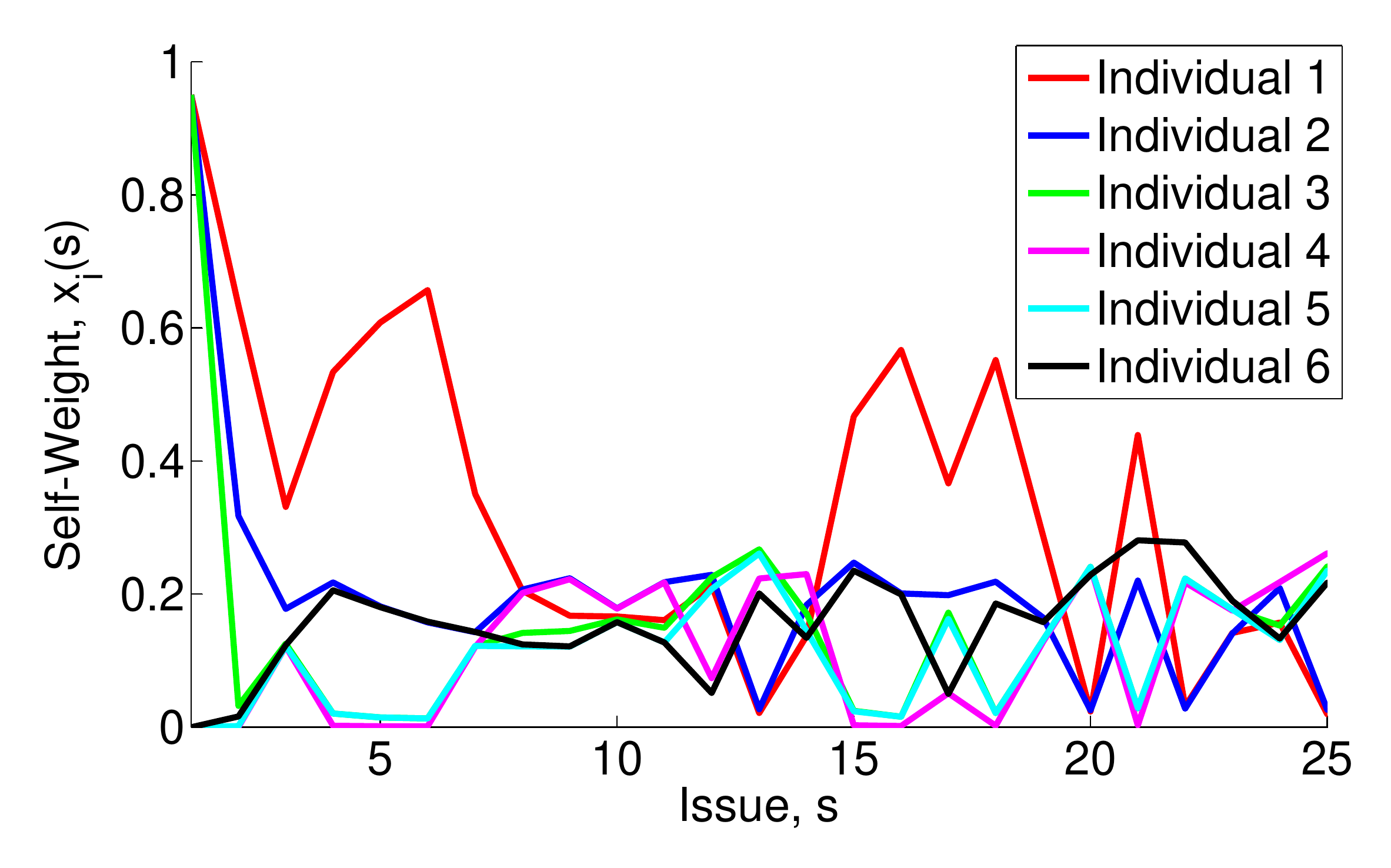}
\caption{Evolution of individuals' social powers $\vect{x}(s)$ for initial condition set $\wh{\vect{x}}(0)$.}
\label{fig:IC_set1}
\end{center}
\end{figure}

\begin{figure}
\begin{center}
\includegraphics[width=0.85\linewidth]{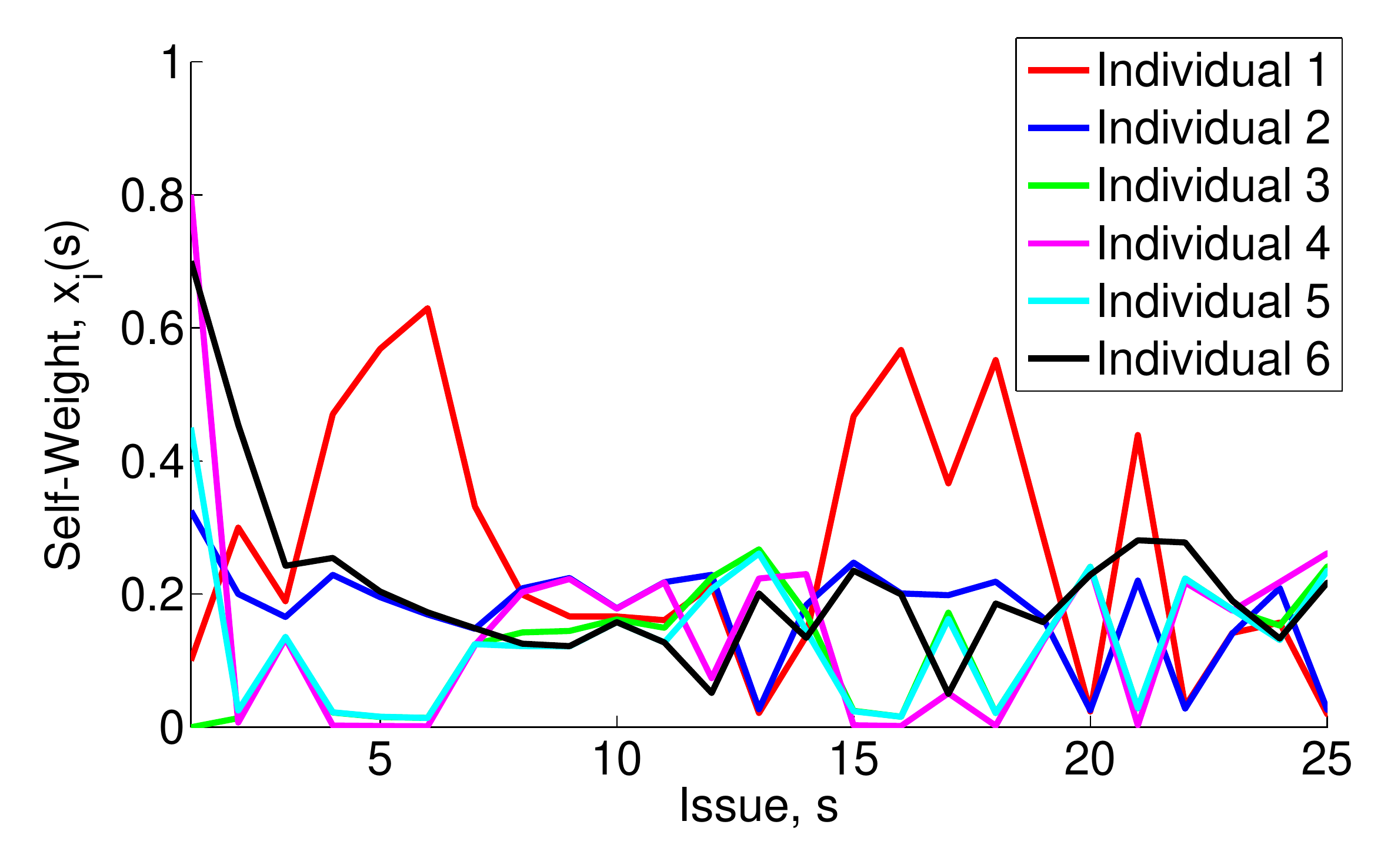}
\caption{Evolution of individuals' social powers $\vect{x}(s)$ for initial condition set $\wt{\vect{x}}(0)$.}
\label{fig:IC_set2}
\end{center}
\end{figure}

\begin{figure}
\begin{center}
\includegraphics[width=0.85\linewidth]{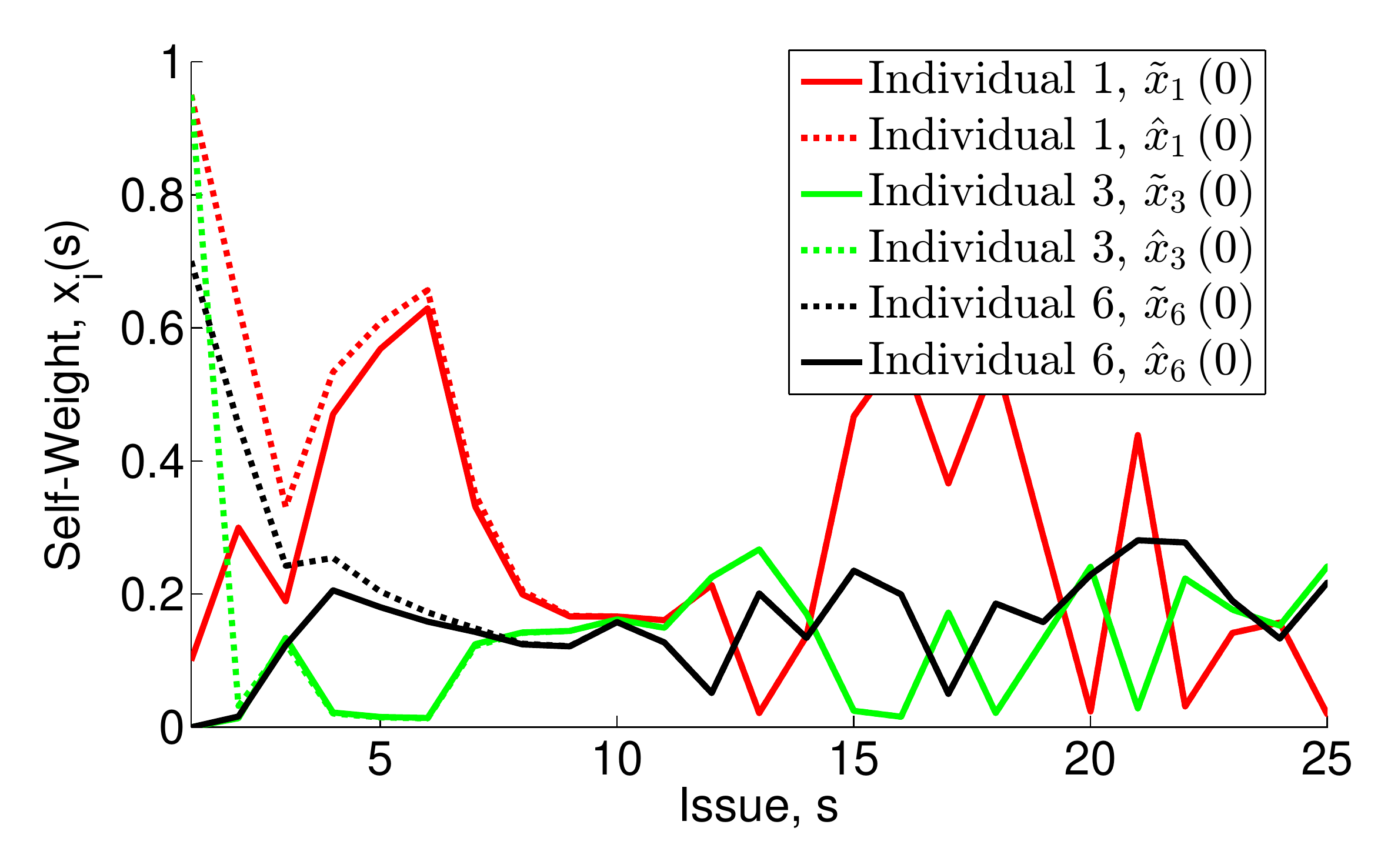}
\caption{Evolution of selected individuals' social powers $x_i(s)$: a comparison of different initial condition sets $\wh{\vect{x}}(0)$ and $\wt{\vect{x}}(0)$.}
\label{fig:IC_compare}
\end{center}
\end{figure}

\section{Conclusion}\label{sec:conclusion}
In this paper, we have presented several novel results on the DeGroot-Friedkin model. For the original model, convergence to the unique equilibrium point has been shown to be exponentially fast. The nonlinear contraction analysis framework allowed for a straightforward extension to dynamic topologies. The key conclusion of this paper is that, according to the DeGroot-Friedkin model, sequential opinion discussion, combined with reflected self-appraisal between any two successive issues, removes perceived (initial) individual social power at an exponential rate. True social power in the limit is determined by the network topology, i.e., interpersonal relationships and their strengths. An upper bound on each individual's limiting social power is computable, depending only on the network topology. 

A number of questions remain. Firstly, we aim to relax the graph topology assumption from strongly connected (i.e., the relative interaction matrix is irreducible) to containing a directed spanning tree (i.e., the relative interaction matrix is reducible). Moreover, one may consider a graph whose union over a set of issues is strongly connected, but for each issue, the graph is not strongly connected. Stubborn individuals (i.e., the Friedkin-Johnsen model) should be incorporated; only partial results are currently available \cite{mirtabatabaei2014stubborn_agent_OD}. Effects of noise and other external inputs should be studied, as well as the concept of personality affecting the reflected self-appraisal mechanism (as mentioned in Remark~\ref{rem:self_reg}).

\appendix  % for no appendix heading

The relative interaction matrices used in the simulation are given by
\begin{align*}
\mat{C}_{1} & = \begin{bmatrix}
0 & 0 & 0 & 0 & 0 & 1 \\
1 & 0 & 0 & 0 & 0 & 0 \\
0 & 1 & 0 & 0 & 0 & 0 \\
0 & 0 & 1 & 0 & 0 & 0 \\
0 & 0 & 0 & 1 & 0 & 0 \\
0 & 0 & 0 & 0 & 1 & 0 \\
\end{bmatrix} \\
\mat{C}_{2} & = \begin{bmatrix}
0 & 0 & 0 & 0 & 1 & 0 \\
0.8 & 0 & 0 & 0 & 0 & 0.2 \\
0 & 0.1 & 0 & 0 & 0 & 0.9 \\
0 & 0 & 1 & 0 & 0 & 0 \\
0 & 0 & 0 & 1 & 0 & 0 \\
0 & 0 & 0 & 0 & 1 & 0 \\
\end{bmatrix}\\
\mat{C}_{3} & = \begin{bmatrix}
0 & 0 & 0 & 0.2 & 0 & 0.8 \\
0.3 & 0 & 0.7 & 0 & 0 & 0 \\
0 & 0 & 0 & 1 & 0.5 & 0 \\
0 & 1 & 0 & 0 & 0 & 0 \\
0.75 & 0 & 0 & 0.25 & 0 & 0 \\
0 & 0 & 0 & 0 & 1 & 0 \\
\end{bmatrix} \\
\mat{C}_{4} & = \begin{bmatrix}
0 & 0 & 0 & 0 & 0.85 & 0.15 \\
1 & 0 & 0 & 0 & 0 & 0 \\
0 & 0.7 & 0 & 0.3 & 0 & 0 \\
0 & 0 & 0.5 & 0 & 0.5 & 0 \\
0 & 0 & 0.9 & 0 & 0 & 0.1 \\
0 & 1 & 0 & 0 & 0 & 0 \\
\end{bmatrix} \\
\mat{C}_{5} & = \begin{bmatrix}
0 & 0.5 & 0 & 0 & 0 & 0.5 \\
0.9 & 0 & 0.1 & 0 & 0 & 0 \\
0.9 & 0 & 0 & 0 & 0 & 0.1 \\
0.9 & 0.1 & 0 & 0 & 0 & 0 \\
0.9 & 0 & 0 & 0.1 & 0 & 0 \\
0.9 & 0 & 0 & 0 & 0.1 & 0 \\
\end{bmatrix}
\end{align*}



\section*{Acknowledgement}
The work of Ye, Anderson, and Yu was supported by the Australian Research Council (ARC) under grants \mbox{DP-130103610} and \mbox{DP-160104500}, and by Data61-CSIRO. The work of Liu and Ba\c{s}ar was supported in part by Office of Naval Research (ONR) MURI Grant N00014-16-1-2710, and in part by NSF under grant CCF 11-11342.
\ifCLASSOPTIONcaptionsoff
  \newpage
\fi



%\bibliographystyle{IEEEtran}
%\begin{thebibliography}{1}
%\end{thebibliography}

\bibliographystyle{IEEEtran}
\bibliography{MYE_ANU}
\begin{IEEEbiography}
[{\includegraphics[width=1in,height=1.25in,clip,keepaspectratio]{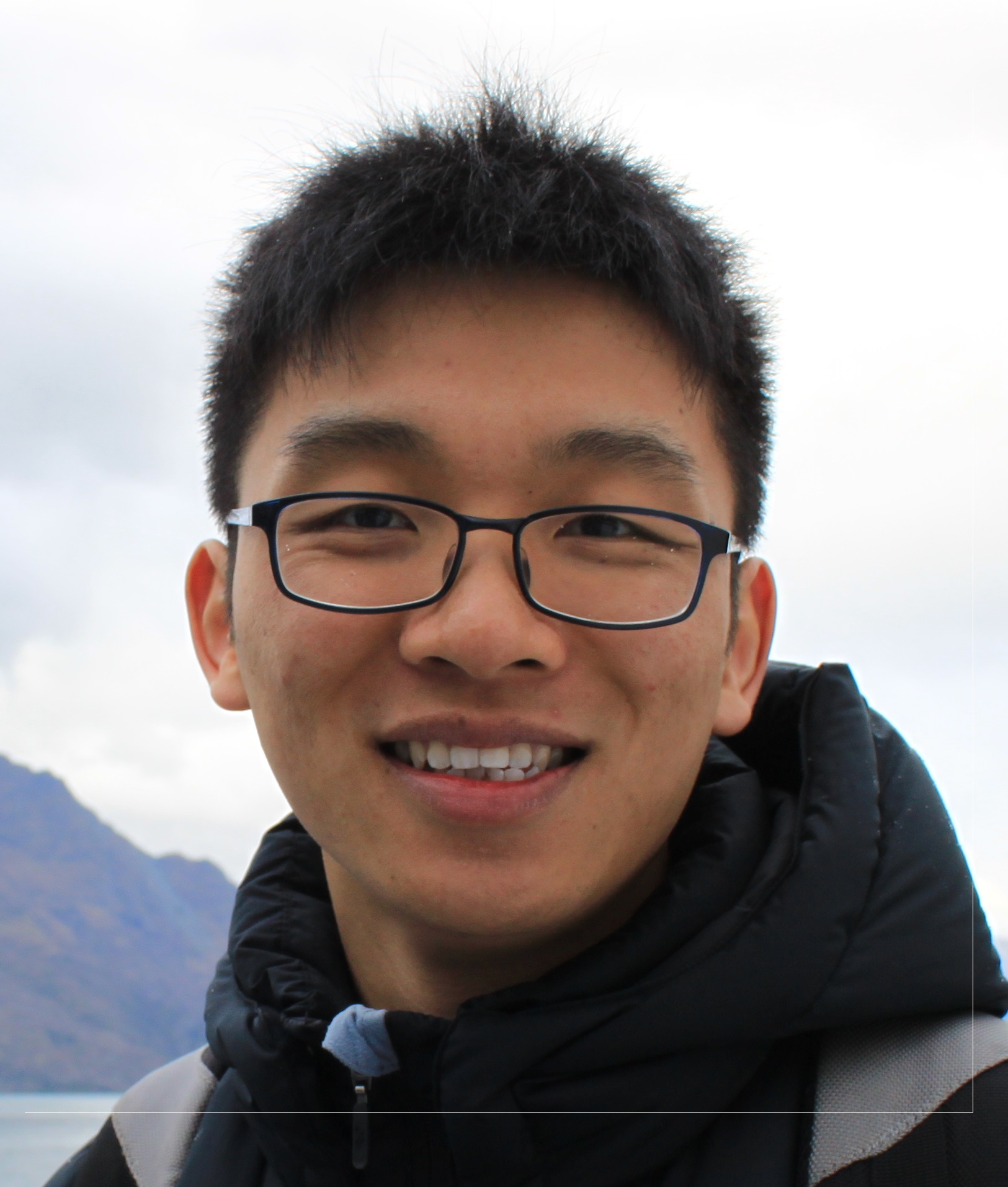}}]{Mengbin Ye}
was born in Guangzhou, China. He received the B.E. degree (with First Class Honours) in mechanical engineering from the University of Auckland, Auckland, New Zealand. He is currently pursuing the Ph.D. degree in control engineering at the Australian National University, Canberra, Australia.

His current research interests include opinion dynamics and social networks, consensus and synchronisation of Euler-Lagrange systems, and localisation using bearing measurements.
\end{IEEEbiography}

\begin{IEEEbiography}
[{\includegraphics[width=1in,height=1.25in,clip,keepaspectratio]{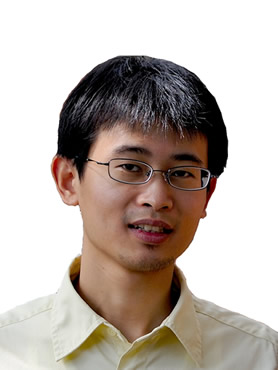}}]{Ji Liu} received the B.S. degree in information engineering from Shanghai Jiao Tong University, Shanghai, China, in 2006, and the Ph.D. degree in electrical engineering from Yale University, 
New Haven, CT, USA, in 2013. He is currently a Postdoctoral Research Associate at the Coordinated Science Laboratory, University of Illinois at Urbana-Champaign, Urbana, IL, USA. 

His current research interests include distributed control and computation, multi-agent systems, social networks, epidemic networks, and power networks.

\end{IEEEbiography}

\begin{IEEEbiography}
[{\includegraphics[width=1in,height=1.25in,clip,keepaspectratio]{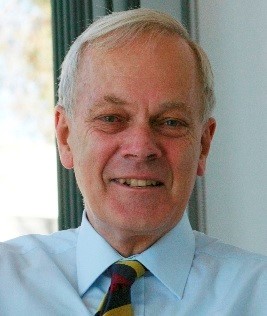}}]{Brian D.O. Anderson} (M'66-SM'74-F'75-LF'07) was born in Sydney, Australia. He received the B.Sc. degree in pure mathematics in 1962, and B.E. in electrical engineering in 1964, from the Sydney University, Sydney, Australia, and the Ph.D. degree in electrical engineering from Stanford University, Stanford, CA, USA, in 1966.

He is an Emeritus Professor at the Australian National University, and a Distinguished Researcher in Data61-CSIRO (previously NICTA) and a Distinguished Professor at Hangzhou Dianzi University. His awards include the IEEE Control Systems Award of 1997, the 2001 IEEE James H Mulligan, Jr Education Medal, and the Bode Prize of the IEEE Control System Society in 1992, as well as several IEEE and other best paper prizes. He is a Fellow of the Australian Academy of Science, the Australian Academy of Technological Sciences and Engineering, the Royal Society, and a foreign member of the US National Academy of Engineering. He holds honorary doctorates from a number of universities, including Universit\'{e} Catholique de Louvain, Belgium, and ETH, Z\"{u}rich. He is a past president of the International Federation of Automatic Control and the Australian Academy of Science. His current research interests are in distributed control, sensor networks and econometric modelling. 
\end{IEEEbiography}

\begin{IEEEbiography}
[{\includegraphics[width=1in,height=1.25in,clip,keepaspectratio]{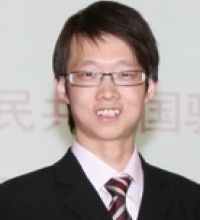}}]{Changbin Yu} received the B.Eng (Hon 1) degree from Nanyang Technological University, Singapore in 2004 and the Ph.D. degree from the Australian National University, Australia, in 2008. Since then he has been a faculty member at the Australian National University and subsequently holding various positions including a specially appointed professorship at Hangzhou Dianzi University.

He had won a competitive Australian Post-doctoral Fellowship (APD) in 2007 and a prestigious ARC Queen Elizabeth II Fellowship (QEII) in 2010. He was also a recipient of Australian Government Endeavour Asia Award (2005) and Endeavour Executive Award (2015), Chinese Government Outstanding Overseas Students Award (2006), Asian Journal of Control Best Paper Award (2006--2009), etc. His current research interests include control of autonomous aerial vehicles, multi-agent systems and human--robot interactions. He is a Fellow of Institute of Engineers Australia, a Senior Member of IEEE and a member of IFAC Technical Committee on Networked Systems. He served as a subject editor for International Journal of Robust and Nonlinear Control and was an associate editor for System \& Control Letters and IET Control Theory \& Applications.
\end{IEEEbiography}

\begin{IEEEbiography}
[{\includegraphics[width=1in,height=1.25in,clip,keepaspectratio]{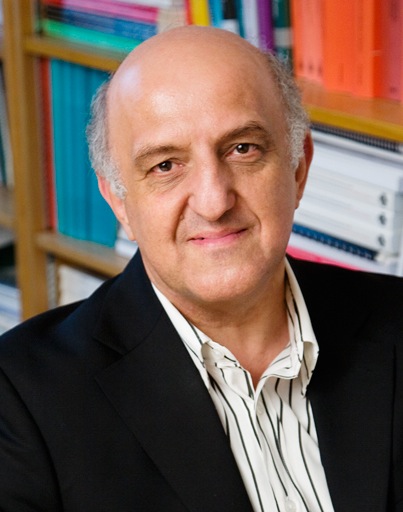}}]{Tamer Ba\c{s}ar} (S'71-M'73-SM'79-F'83-LF'13) is with the University of Illinois at Urbana-Champaign (UIUC), where he holds the academic positions of 
Swanlund Endowed Chair;    
Center for Advanced Study Professor of  Electrical and Computer Engineering; 
Research Professor at the Coordinated Science
Laboratory; and Research Professor  at the Information Trust Institute. 
He is also the Director of the Center for Advanced Study.
He received B.S.E.E. from Robert College, Istanbul,
and M.S., M.Phil, and Ph.D. from Yale University. He is a member of the US National Academy
of Engineering,  the European Academy of Sciences, and Fellow of IEEE, IFAC and SIAM, and has served as president of IEEE CSS,
ISDG, and AACC. He has received several awards and recognitions over the years, including  the IEEE Control Systems Award, 
the highest awards of IEEE CSS, IFAC, AACC, and ISDG, and a number of international honorary
doctorates and professorships. He has over 800 publications in systems, control, communications, networks,
and dynamic games, including books on non-cooperative dynamic game theory, robust control,
network security, wireless and communication networks, and stochastic networked control. He was
the Editor-in-Chief of Automatica between 2004 and 2014, and is currently  editor of several book series. His current research interests
include stochastic teams, games, and networks; security; and cyber-physical systems.
\end{IEEEbiography}


%\vfill


\end{document}